\pdfoutput=1

\def\final{0}

\def\mnotes{0}
\def\mydebug{0}

\ifnum\final=1
  \documentclass[10pt,conference,letterpaper]{IEEEtran}
\else
   \documentclass[11pt]{article}
	\usepackage{fullpage}
   \usepackage{amssymb, amsfonts,amsmath, amsthm}
\fi
\def\colorson{0}

\usepackage{graphicx}
\usepackage{color, url}

\usepackage{wrapfig}

\advance \topmargin by -\headheight

\advance \topmargin by -\headsep

\ifnum\colorson=1
\newenvironment{todo}{\noindent
\sf \footnotesize \textcolor{blue}{To go here}:
\begin{CompactItemize}\color{blue}}
{\color{black}\end{CompactItemize}\rm \normalsize}
\else

\fi

\ifnum\final=0

\theoremstyle{definition}

\theoremstyle{remark}

\newtheorem{myremark}{Remark} [section]
\newenvironment{remark}{\begin{myremark}}{$\diamondsuit$\end{myremark}}

\newtheorem{myexample}{Example}

\fi

\newtheorem{theorem}{Theorem}[section]
\newtheorem{corollary}[theorem]{Corollary}
\newtheorem{lemma}[theorem]{Lemma}
\newtheorem{definition}{Definition}[section]

\newtheorem{proposition}[theorem]{Proposition}

\newtheorem{fact}[theorem]{Fact}

\newtheorem{observation}{Observation}

\usepackage{amsfonts,amsmath, amsopn, amssymb, epsfig, graphicx, multirow, url, verbatim}
\usepackage{color,bm}

\newcommand{\D}{{\cal D}}

\ifnum\mnotes=0
\newcommand{\mnote}[1]{}
\else
\newcounter{mynotes}
\newcommand{\mnote}[1]{\addtocounter{mynotes}{1}{{\bf !}}
{\scriptsize  {\arabic{mynotes}.\ {\sf \textcolor{red}{#1}}}}}
\fi

\ifnum\final=0

\newcommand{\divnote}[1]{}
\newcommand{\grinote}[1]{}
\newcommand{\magnote}[1]{}
\newcommand{\granote}[1]{}

\else

\newcommand{\divnote}[1]{\textbf{Div: #1}}
\newcommand{\grinote}[1]{\textbf{Gri: #1}}
\newcommand{\magnote}[1]{\textbf{Mag: #1}}
\newcommand{\granote}[1]{\textbf{Gra: #1}}

\fi

\newcommand{\tabref}[1]{Table~\ref{tab:#1}}

\newcommand{\ignore}[1]{}

\newcommand{\eps}{\varepsilon}

\DeclareMathOperator{\Var}{Var}
\DeclareMathOperator{\Lap}{Lap}

\usepackage{algorithm, algorithmic}
\usepackage[algo2e,linesnumbered,ruled,vlined]{algorithm2e}

\newenvironment{myproof}{\begin{proof}
}{\ifnum\final=1
{\hfill\qed}
\fi
\end{proof}}
\newenvironment{proofof}[1]{\begin{myproof}[\ifnum\final=0 Proof \fi of {#1}]
}{\end{myproof}}

\newcommand{\plots}{.}

\usepackage{times}
\usepackage{subfigure}
\newcommand{\eat}[1]{}
\newcommand{\transpose}{\ensuremath{^{T}}}
\DeclareMathOperator{\Cov}{Cov}
\DeclareMathOperator{\diag}{diag}
\newcommand{\E}{\mathsf{E}}
\newcommand{\F}{\mathcal{F}}

\begin{document}

\ifnum\final=1
\newcommand{\qed}{\hfill \rule{1ex}{1ex}}
\title{Accurate and Efficient Private Release of Datacubes and
Contingency Tables}
\author{
\authorblockN{Grigory Yaroslavtsev}
\authorblockA{Pennsylvania State University\\
grigory@cse.psu.edu}
\and
\authorblockN{Graham Cormode ~~ Cecilia M. Procopiuc ~~ Divesh Srivastava}
\authorblockA{AT\&T Labs -- Research
\\\{graham,magda,divesh\}@research.att.com}
}
\maketitle
\else
\title{Accurate and Efficient Private Release of Datacubes}
\author{Graham Cormode\thanks{AT\&T Labs -- Research, Florham Park, NJ 07932., {\tt graham@research.att.com}} \and
Cecilia M. Procopiuc\thanks{AT\&T Labs -- Research, Florham Park, NJ 07932., {\tt magda@research.att.com}} \and Divesh
Srivastava\thanks{AT\&T Labs -- Research, Florham Park, NJ 07932., \tt divesh@research.att.com} \and Grigory
Yaroslavtsev\thanks{Pennsylvania State University, University Park, PA 16802., {\tt grigory@grigory.us}}}
\maketitle
\fi

\begin{abstract}
A central problem in releasing aggregate information about
sensitive data is to do so accurately
while providing a privacy guarantee on the output.
Recent work focuses on the class of {\em linear queries}, which include basic
counting queries, data cubes, and contingency tables. The goal is to
maximize the utility of their output, while giving a rigorous privacy
guarantee.
Most results follow a common template: pick a ``strategy'' set of linear
queries to apply to the data, then use the noisy answers to these
queries to reconstruct the queries of interest.
This entails either picking a strategy set that is hoped to be good
for the queries, or performing a costly search over the space
of all possible strategies.

In this paper, we propose a new approach that balances accuracy and
efficiency:
  we show how to improve the accuracy of a given query set by
  answering some strategy queries more accurately than others.
This leads to an efficient optimal noise allocation for many popular
strategies,
including wavelets, hierarchies, Fourier coefficients and more.
For the important case of marginal queries we show that this strictly improves on previous
methods, both analytically and empirically.
Our results also extend to ensuring that the returned query answers
are consistent with an (unknown) data set at minimal extra cost in
terms of time and noise.
\end{abstract}

\section{Introduction}\label{intro}

The long-term goal of much work in data privacy is to enable the
release of information that accurately captures the behavior of an
input data set, while preserving the privacy of individuals described
therein. 
There are two central, interlinked questions to address around this
goal: what privacy properties should the transformation process
possess, and how can we ensure that the output is useful for subsequent
analysis and processing? 
The model of Differential Privacy has lately gained broad acceptance as a
criterion for private data release~\cite{Dwork:06,DMNS06}. 
There are now multiple different methods which achieve
Differential Privacy over different data types~\cite{BCDKST07,BLR08,CPSSY12,DWHL11,HLS10,HR10,HRMS10,LHRMG10,LM11,XWG10}. 
Some provide a strong utility guarantee, while others demonstrate
their utility via empirical studies. 
These algorithms also vary from the highly practical, to taking time 
exponential in the data size.

The output of the data release should be compatible with existing
tools and processes in order to provide usable results. 
The model of contingency tables is universal, in that any relation can
be represented exactly in this form.
That is, the contingency table of a dataset over a
subset of attributes contains, for each possible attribute combination,
the number of tuples that occur in the data with that set of attribute
values. 
In this paper, we call such a contingency table the {\em marginal} of
the database over the respective subset of attributes. 
The set of all possible marginals for a relation is captured
by the data cube. 
Contingency tables and the data cube in turn are examples of a more
general class of linear queries, 
i.e., each query is a linear combination of the entries of the contingency
table over all attributes in the input relation. 

There has been much interest in providing methods to answer such
linear queries with privacy guarantees. 
In this paper, we argue that these all fit within a general framework: 
answer some set of queries $S$ over the data (not necessarily the set that
was requested), with appropriate noise added to provide the privacy, 
then use the answers to answer the given queries 
A limitation of prior work is that it applies {\em uniform noise}
to the answers to $S$: the same magnitude of noise is added to each query. 
However, it turns out that the accuracy can be much improved by
using {\em non-uniform noise}: using different noise for
each answer, while providing the same overall guarantee. 
The main contribution of this paper is to provide a full formal
understanding of this problem and the role that non-uniform noise can
play. 

\noindent
{\bf Example.} 
Figure~\ref{toyextable} 
shows a table with 3 binary attributes $A$, $B$ and $C.$
As in prior work~\cite{LHRMG10}, we think of a database $\D$ as an
$N$-dimensional vector $x\in\mathbb R^N$, where $N$ is the domain size
of $\D;$ i.e., if $\D$ has attributes  
$A_1,\ldots,A_d,$ then $N=\Pi_{i=1}^d |A_i|.$ 
We linearize the domain of $\D$, so that each index position $i$, 
$1\leq i\leq N$,  corresponds to a unique combination $\alpha$ 
of attribute values, and $x_i$ is the number of tuples in 
$\D$ that have values $\alpha$. 
In
Figure~\ref{toyextable},
we linearized the domain in the order 
$000, 001, \ldots, 111.$ 
Here, position $i=2$ corresponds to the combination of values 
$\alpha = 001$. 
Thus, $x_2=2$ since $\D$ contains two tuples (1 and 4) with these values. 

\ifnum\final=1
\begin{figure*}
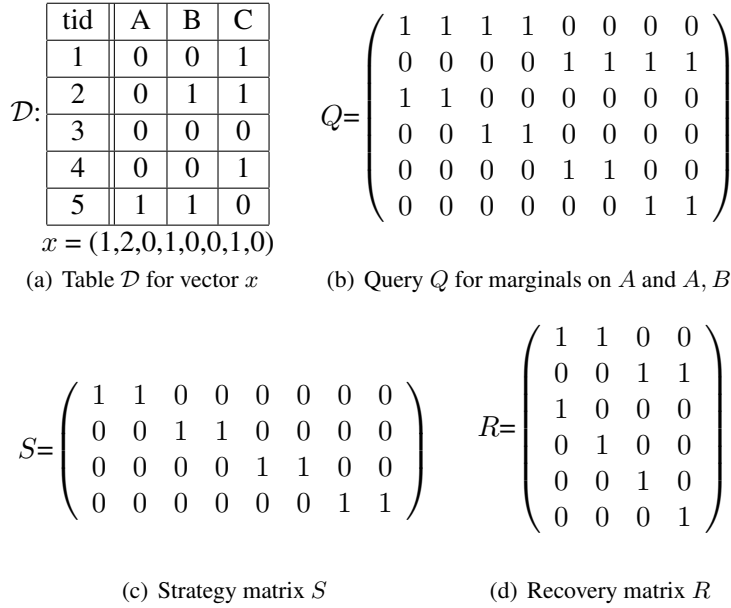

\small
\subfigure[Table $\D$ for vector $x$]{
\label{toyextable}
\begin{tabular}[b]{c@{\hspace*{-0.02mm}}c@{\hspace*{-0.02mm}}c}
$\D$ & : & \begin{tabular}{|c||c|c|c|}
\hline
tid & A & B & C\\
\hline
1 & 0 & 0 & 1\\
\hline
2 & 0 & 1 & 1\\
\hline
3 & 0 & 0 & 0\\
\hline
4 & 0 & 0 & 1\\
\hline
5 & 1 & 1 & 0\\
\hline
\end{tabular}\\[9mm]
& & $x$ = (1,2,0,1,0,0,1,0)
\end{tabular} 
}%
\subfigure[Query $Q$ for marginals on $A$ and $A,B$]{
\label{toyexquery}
\begin{tabular}[b]{c@{\hspace*{-0.02mm}}c@{\hspace*{-0.02mm}}c}
$Q$ & = & $\left( \begin{array}{cccccccc}
1 & 1 & 1 & 1 & 0 & 0 & 0 & 0\\
0 & 0 & 0 & 0 & 1 & 1 & 1 & 1\\
1 & 1 & 0 & 0 & 0 & 0 & 0 & 0\\
0 & 0 & 1 & 1 & 0 & 0 & 0 & 0\\
0 & 0 & 0 & 0 & 1 & 1 & 0 & 0\\
0 & 0 & 0 & 0 & 0 & 0 & 1 & 1\\
\end{array} \right)$\\[9mm]
& &\\ 
\end{tabular} 
}%
\subfigure[Strategy matrix $S$]{
\label{toyexstrategy}
\begin{tabular}[b]{c@{\hspace*{-0.02mm}}c@{\hspace*{-0.02mm}}c}
$S$ & = & $\left( \begin{array}{cccccccc}
1 & 1 & 0 & 0 & 0 & 0 & 0 & 0\\
0 & 0 & 1 & 1 & 0 & 0 & 0 & 0\\
0 & 0 & 0 & 0 & 1 & 1 & 0 & 0\\
0 & 0 & 0 & 0 & 0 & 0 & 1 & 1\\
\end{array} \right)$\\[9mm]
& &\\
\end{tabular} 
}%
\subfigure[Recovery matrix $R$]{
\label{toyexrecovery}
\begin{tabular}[b]{c@{\hspace*{-0.02mm}}c@{\hspace*{-0.02mm}}c}
$R$ & = & $\left( \begin{array}{cccc}
1 & 1 & 0 & 0\\
0 & 0 & 1 & 1\\
1 & 0 & 0 & 0\\
0 & 1 & 0 & 0\\
0 & 0 & 1 & 0\\
0 & 0 & 0 & 1\\
\end{array} \right)$\\[9mm]
& &\\
\end{tabular}}
\caption{Example contingency table, query matrix, with strategy and recovery matrices.}
\label{toyex}
\end{figure*}

\else

\begin{figure*}
 \centering
\subfigure[Table $\D$ for vector $x$]{
\label{toyextable}
\begin{tabular}[b]{c@{\hspace*{-0.02mm}}c@{\hspace*{-0.02mm}}c}
$\D$ & : & \begin{tabular}{|c||c|c|c|}
\hline
tid & A & B & C\\
\hline
1 & 0 & 0 & 1\\
\hline
2 & 0 & 1 & 1\\
\hline
3 & 0 & 0 & 0\\
\hline
4 & 0 & 0 & 1\\
\hline
5 & 1 & 1 & 0\\
\hline
\end{tabular}\\[9mm]
& & $x$ = (1,2,0,1,0,0,1,0)
\end{tabular} 
}%
\subfigure[Query $Q$ for marginals on $A$ and $A,B$]{
\label{toyexquery}
\begin{tabular}[b]{c@{\hspace*{-0.02mm}}c@{\hspace*{-0.02mm}}c}
$Q$ & = & $\left( \begin{array}{cccccccc}
1 & 1 & 1 & 1 & 0 & 0 & 0 & 0\\
0 & 0 & 0 & 0 & 1 & 1 & 1 & 1\\
1 & 1 & 0 & 0 & 0 & 0 & 0 & 0\\
0 & 0 & 1 & 1 & 0 & 0 & 0 & 0\\
0 & 0 & 0 & 0 & 1 & 1 & 0 & 0\\
0 & 0 & 0 & 0 & 0 & 0 & 1 & 1\\
\end{array} \right)$\\[9mm]
& &\\ 
\end{tabular} 
}%

\subfigure[Strategy matrix $S$]{
\label{toyexstrategy}
\begin{tabular}[b]{c@{\hspace*{-0.02mm}}c@{\hspace*{-0.02mm}}c}
$S$ & = & $\left( \begin{array}{cccccccc}
1 & 1 & 0 & 0 & 0 & 0 & 0 & 0\\
0 & 0 & 1 & 1 & 0 & 0 & 0 & 0\\
0 & 0 & 0 & 0 & 1 & 1 & 0 & 0\\
0 & 0 & 0 & 0 & 0 & 0 & 1 & 1\\
\end{array} \right)$\\[9mm]
& &\\
\end{tabular} 
}%
\subfigure[Recovery matrix $R$]{
\label{toyexrecovery}
\begin{tabular}[b]{c@{\hspace*{-0.02mm}}c@{\hspace*{-0.02mm}}c}
$R$ & = & $\left( \begin{array}{cccc}
1 & 1 & 0 & 0\\
0 & 0 & 1 & 1\\
1 & 0 & 0 & 0\\
0 & 1 & 0 & 0\\
0 & 0 & 1 & 0\\
0 & 0 & 0 & 1\\
\end{array} \right)$\\[9mm]
& &\\
\end{tabular}}
\caption{Example contingency table, query matrix, with strategy and recovery matrices.}
\label{toyex}
\end{figure*}
\fi

\eat{
\begin{figure*}
{\scriptsize
\begin{tabular}{cccc}
\begin{tabular}{c@{\hspace*{-0.02mm}}c@{\hspace*{-0.02mm}}c}
$\D$ & : & \begin{tabular}{|c||c|c|c|}
\hline
tid & A & B & C\\
\hline
1 & 0 & 0 & 1\\
\hline
2 & 0 & 1 & 1\\
\hline
3 & 0 & 0 & 0\\
\hline
4 & 0 & 0 & 1\\
\hline
5 & 1 & 1 & 0\\
\hline
\end{tabular}\\[9mm]
& & $x$ = (1,2,0,1,0,1,0)\\
\end{tabular} 
&
\begin{tabular}{c@{\hspace*{-0.02mm}}c@{\hspace*{-0.02mm}}c}
$Q$ & = & $\left( \begin{array}{cccccccc}
1 & 1 & 1 & 1 & 0 & 0 & 0 & 0\\
0 & 0 & 0 & 0 & 1 & 1 & 1 & 1\\
1 & 1 & 0 & 0 & 0 & 0 & 0 & 0\\
0 & 0 & 1 & 1 & 0 & 0 & 0 & 0\\
0 & 0 & 0 & 0 & 1 & 1 & 0 & 0\\
0 & 0 & 0 & 0 & 0 & 0 & 1 & 1\\
\end{array} \right)$\\[9mm]
& &\\ 
\end{tabular} 
&
\begin{tabular}{c@{\hspace*{-0.02mm}}c@{\hspace*{-0.02mm}}c}
$S$ & = & $\left( \begin{array}{cccccccc}
1 & 1 & 0 & 0 & 0 & 0 & 0 & 0\\
0 & 0 & 1 & 1 & 0 & 0 & 0 & 0\\
0 & 0 & 0 & 0 & 1 & 1 & 0 & 0\\
0 & 0 & 0 & 0 & 0 & 0 & 1 & 1\\
\end{array} \right)$\\[9mm]
& &\\
\end{tabular} 
&
\begin{tabular}{c@{\hspace*{-0.02mm}}c@{\hspace*{-0.02mm}}c}
$R$ & = & $\left( \begin{array}{cccc}
1 & 1 & 0 & 0\\
0 & 0 & 1 & 1\\
1 & 0 & 0 & 0\\
0 & 1 & 0 & 0\\
0 & 0 & 1 & 0\\
0 & 0 & 0 & 1\\
\end{array} \right)$\\[9mm]
& &\\
\end{tabular}\\[1.7cm]
(a) & (b) & (c) & (d)\\
\end{tabular}
}
\caption{(a) Table $\D$ and corresponding vector $x$; (b) query $Q$ computes the marginals on $A$ and on $A,B$;
(c) strategy matrix $S$; (d) recovery matrix $R.$\label{toyex}}
\end{figure*}
}

Suppose that we want to compute two marginals over $\D$: the marginal
over $A$, and the marginal over $A,B$. 
The query marginals can be represented as a matrix $Q$,
as depicted in 
Figure~\ref{toyexquery}, so that the answer is $Qx$:
The first two rows compute the marginal over $A$; i.e., the first row
is the linear query that counts all tuples  $t$ with $t.A=0$; while
the second row counts all tuples with $t.A=1.$ 
Similarly, the third row counts all tuples with $t.A=0$ and $t.B=0$,
and so on. 
Differentially private mechanisms answer $Q$ in the
form of $y=Qx+\tau$,
where $\tau$ is a random vector whose distribution provides a certain
level of privacy.
The error of the answer is generally defined as the variance $\Var(y)$
\cite{LHRMG10,XWG10}. 

For example, one way to provide $\eps$-differential privacy 
adds uniform noise to each answer.  
Based on the structure of $Q$ in Figure~\ref{toyexquery}, we can add noise with 
variance $\frac{8}{\eps^2}$ to each answer; see details in Section~\ref{defs}).
Over the six queries, the sum of variances is  
$\frac{48}{\eps^2}$. 
However, we can do better with a non-uniform approach. 
For example, we can add noise with variance
$2(\frac{9}{4\eps})^2$ to the answers for the first two rows of $Q$, 
and noise with variance $2(\frac{9}{5\eps})^2$ to the remaining four
answers, and still provide $\eps$-differential privacy. 
The sum of the six variances is then 
$2\cdot2(\frac{9}{4\eps})^2 + 4 \cdot 2(\frac{9}{5\eps})^2 = 46.17/\eps^2$.
We can improve this even further by changing how we answer the
queries:  we can answer the first query $Q_1$ by taking 
half of the first answer, and adding half of the third and fourth
answers. 
The resulting variance of $Q_1$ is
\[\textstyle \frac{1}{4} \cdot 2 \big(\frac{9}{4\eps}\big)^2 + 
\frac{1}{4} \cdot 2\big(\frac{9}{5\eps}\big)^2 + 
\frac{1}{4} \cdot 2\big(\frac{9}{5\eps}\big)^2 = 5.77/\eps^2.\] 
Similar tricks yield the same variance for all other answers, so the sum
of all six variances is now $34.6/\eps^2$, a 28\% reduction over the uniform approach.
\qed




\eat{
For downstream compatibility, it is sometimes necessary to enforce
additional {\em consistency} constraints on the released data.
For example, in Figure~\ref{toyex}, one may wish to enforce the
condition that the noisy count for $t.A=0$ is  
equal to the noisy count for $(t.A=0\mbox{ and }t.B=0)$ plus the noisy
count for $(t.A=0\mbox{ and }t.B=1)$.
Additional consistency conditions can be added, such as enforcing
non-negative counts. 
For brevity, we omit a detailed discussion of such extra conditions. 
As we discuss later, consistent counts can be achieved
while minimizing the error $\Var(y)$. 
}

This example shows that we can significantly improve the accuracy of
our answers while preserving the same level of privacy by adopting
non-uniform noise and careful combination of intermediate answers to
give the final answer. 
Yet further improvement can result by choosing a different set of
queries to obtain noisy answers to. 
The problem we address in this paper is how to use these techniques to
efficiently and accurately provide answers to such queries $Q$ that meet the
differential privacy guarantee.  
This captures the core problems of releasing data cubes, contingency
tables and marginals. 
Our results are more general, as they apply to arbitrary
sets of linear queries $Q$, but our focus is on these important
special cases.  
We also discuss how to additionally ensure that the answers meet
certain consistency criteria. 
Next, we study how existing techniques can be applied to this problem,
and discuss their limitations. 

\smallskip
\noindent
{\bf The Strategy/Recovery approach.}
Mechanisms for minimizing the error of linear counting queries under
differential privacy have attracted a lot of attention. 
Work in the theory
community~\cite{BLR08, GRU12, HLS10, HR10, HT10, RR10} has focused on 
providing the best bounds
on noise for an arbitrary set of such queries, in both the online and
offline setting. 
However, these mechanisms are rarely practical for large databases with
moderately high dimensionality: they can scale exponentially with the size
of their input. 

Work in database research has aimed to deliver methods that scale to realistic data sizes. 
Much of this work builds on basic primitives in differential privacy
such as adding appropriately scaled noise to a numeric quantity from a
specific random distribution (see Section~\ref{defs}).
Repeating this process for multiple different quantities, and
reasoning about how the privacy guarantees compose, it is possible to
ensure that the full output meets the privacy definition.
The goal is then to minimize the error
introduced into the query answers (as measured by their variance)
while satisfying the privacy conditions. 

Given this outline, we observe that the bulk of methods using noise
addition fit into a two-step framework that we dub the 
`strategy/recovery' approach:

\begin{itemize}
\item {\bf Step 1.} Find a {\em strategy matrix} $S$ and compute the
  vector $z = Sx + \nu$, where $\nu$ is a random {\em noise vector}
  drawn from an appropriate distribution. 
Then $z$ is the differentially private answer to the queries represented by $S$.
\item {\bf Step 2.} Compute a {\em recovery matrix} $R$, such that $Q=RS$. 
Return $y=Rz$ as the differentially private answer to the queries $Q.$ 
The variance $\Var(y)$ is often used as an error measure for the approach. 
\end{itemize} 

We show this method schematically in Figure~\ref{overview1}.
For example, Figures~\ref{toyexstrategy} and \ref{toyexrecovery} show
a possible choice of matrices $S$ and $R$ for the query matrix $Q$ in
Figure~\ref{toyexquery}.  
In this case, the strategy $S$ computes the marginal on $A,B$; Step 1
above adds random noise independently to all cells in this marginal. 
The recovery $R$ computes the marginal on $A$ by aggregating the
corresponding noisy cells from the marginal on $A,B$ 
(the first two rows of $R$), 
and also outputs the marginal on $A,B$ (the last four rows of $R$). 

We now show how prior work fits into this approach. 
In many cases, the first step directly picks a fixed matrix for $S$, 
by arguing that this is suitable for a particular class of queries
$Q$. 
For example, when setting $S=I$ (hence $R=Q$), the approach 
computes a set of noisy counts $\tilde{x}_i$ by
adding Laplace noise independently to each $x_i$. 
The answer to any query matrix $Q$ is computed over these noisy
counts, i.e., $y=Q\tilde{x}$; this model was analyzed 
in~\cite{BCDKST07}. 
By contrast, when $S=Q$ (and $R=I$), 
as discussed in~\cite{Dwork:06}, 
the approach adds noise to the result
of each query in $Q$, i.e., $y=Qx + \nu$. 

\begin{figure}
\centering
\includegraphics[width=0.4\textwidth]{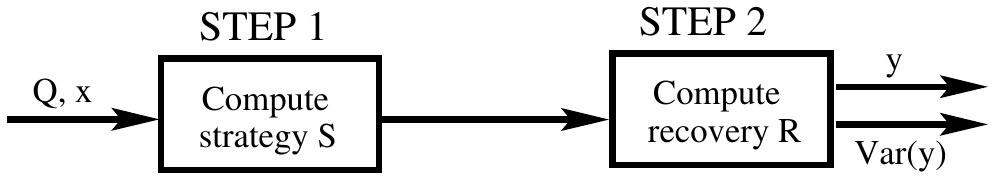}
\caption{Framework of prior work.}
\label{overview1}
\end{figure}

Several more sophisticated strategies have been designed, with the
goal of minimizing the error $\Var(y)$ for various query workloads. 
When $Q$ consists of low-dimensional range queries, \cite{XWG10}
proposes $S$ to be the wavelet transform, 
while~\cite{HRMS10} studies the strategy $S$ corresponding to a
hierarchical structure over $x$. 
However, as shown in~\cite{LHRMG10}, neither of these strategies is
particularly accurate for other types of queries. 
For marginals, \cite{BCDKST07} chooses $S$ to be the Fourier transform
matrix, and~\cite{DWHL11} employs a clustering algorithm over the
queries to compute $S$. 
Figures~\ref{toyexstrategy} and \ref{toyexrecovery} depict the output computed
via~\cite{DWHL11} on query matrix $Q$ (Figure~\ref{toyexquery}).
Other work has suggested the use of random projections as the strategy
matrix, connecting to the area of sparse recovery~\cite{CPST12,LZWY11}.
Many of these choices are relatively fast: 
that is, $S$ and $S^{-1}$ can be applied to a vector of length $N$ in time
$O(N)$ or $O(N\log N)$ in the case of wavelet and Fourier transforms,
respectively. 
This is important, since real data can have large values of $N$, and
so asymptotically higher running time may not be practical. 
A limitation of~\cite{DWHL11} is that the clustering step is very
expensive, limiting the scalability of the approach.

An important technical distinction for the strategy/recovery approach is whether
or not the strategy $S$ is invertible. 
If it is (e.g., when $S$ is the Fourier or wavelet transform), then
the recovery matrix $R=QS^{-1}$ is unique,
and the query answer $y$ is
guaranteed to be consistent (see Definition~\ref{def-consist}). 
Then the error measure $\Var(y)$ depends only on $S$ (and $Q$). 
However, if $S$ is not invertible, then there can be many choices for
$R$, and $\Var(y)$ depends on both $S$ and $R$. 
The optimal recovery $R$ that minimizes $\Var(y)$ (for
a fixed $S$)  can be computed via the least squares
method~\cite{HRMS10,LHRMG10}
and $\Var(y)$ has a closed-form expression as a function of $S$. 
Using this fact, Li et al.~\cite{LHRMG10} study the following
optimization problem: {\em Given queries $Q$ and a formula for
  $\Var(y)$ as a function of $S$, compute the strategy $S$ that
  minimizes $\Var(y)$}. 
This is a tough optimization, since the search is over
all possible strategy matrices $S$. 
Their {\em matrix mechanism} uses a rank-constrained 
semidefinite program (SDP) to compute the optimal
$S$.
Solving this SDP is very costly as a function of $N$, making it
impractical for data with more than a few tens of entries. 

In summary, the search for a strategy matrix $S$ is currently done
either by picking one that we think is likely to be ``good'' for queries $Q$,
or by solving an SDP, which is impractical even for moderate size
problems. 

%

\smallskip
\noindent
{\bf Our Contributions.}
Most of the prior approaches discussed above use the
uniform ``noise budgeting'' strategy, i.e., each 
value $\nu_i$ of the noise vector is (independently) drawn from the
same random distribution. 
The scaling parameter of this distribution depends on the desired
privacy guarantee $\eps$, as well as the ``sensitivity'' of the
strategy matrix $S$ (see Section~\ref{defs}). 

In the extended version of~\cite{LHRMG10}, the authors prove that any
non-uniform noise budgeting strategy can be reduced to a uniform
budgeting strategy by scaling the rows of $S$ with different factors. 
However, computing the optimal scaling factors this way is
impractical, as it requires solving an SDP.  
The only efficient method for computing non-uniform noise budgets
we are aware of applies to the
special case when $Q$ is a range query workload~\cite{CPSSY12}. 
There, $S$
corresponds to a multi-dimensional hierarchical decomposition, 
and recovery $R$ corresponds to the greedy range decomposition. The resulting budgeting is not
always optimal.

In this paper we show how to compute the {\em optimal} noise budgets in
time at most linear in the sizes of $R$ and $S$, for a large class of queries
$Q$ (including marginal queries), and for most of the matrices $S$
considered in prior work. 
This includes the Fourier transform, the wavelet transform, the
hierarchical structure over $x$, and any strategy consisting of a set
of marginals (in particular, the clustering strategy of~\cite{DWHL11}). 

The overall framework introduced is depicted in
Figure~\ref{overview2}: Given strategy matrix $S$ and recovery matrix 
$R$, we compute optimal noise budgets $\eps_i$ for each query, and draw each 
noise value $\nu_i$ from a random distribution that depends on $\eps_i$ (Step 
2). We then derive a new recovery matrix $R$ that minimizes
$\Var(y)$ (Step 3), for the noise budgets computed in Step 2.

\begin{figure}[t]
\centering
\includegraphics[width=0.49\textwidth]{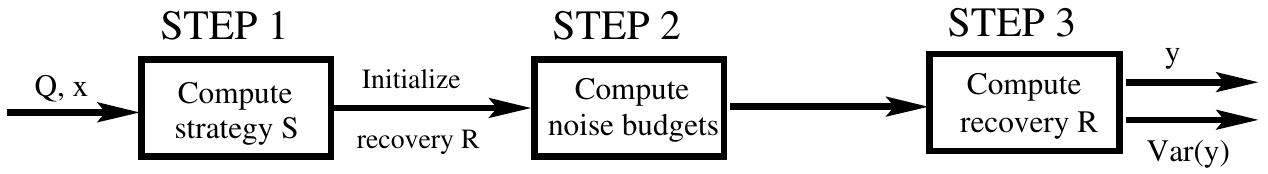}
\caption{Our proposed framework.}
\label{overview2}
\end{figure}

The most general approach would be 
to provide a mathematical formulation for the following 
global optimization problem: 
{\em Given the query matrix $Q$, compute the strategy $S$, the recovery $R$, 
and the noise budgets $\eps_i$ that minimize $\Var(y)$.} 
However, this problem essentially reduces to that addressed by the
matrix mechanism~\cite{LHRMG10}, and requires solving an SDP. 

Instead, we study how to efficiently solve optimization problems where two
out of the three parameters $S$, $R$ and $\{\eps_i\}_i$  are fixed. 
In Section~\ref{noisebudget}, we solve the optimization problem: 
{\em Given a decomposition of query matrix $Q$ into strategy $S$ and
  recovery $R$, compute the optimal noise budgets $\eps_i$ that
  minimize $\Var(y)$}.  
We provide a formula for $\Var(y)$, as a function of $S$ and $R$. 
In  Section~\ref{ols}, we apply the generalized least squares
method to solve the following problem: {\em Given the query matrix
  $Q$, the strategy $S$, and the noise budgets $\eps_i$, compute the
  recovery $R$ that minimizes $\Var(y)$.}  
Following the steps in this framework provides efficient algorithms with
low error. A faster alternative computes a consistent output $y$ of Step 3
with small (but non-optimal) error; see Sections~\ref{subsec:fastconsist} 
and \ref{sec:consistency}. Our approach strictly 
improves over the previous result from~\cite{DWHL11}.


\eat{
In principle, we could then compute a new $S$, given $R$ and $\eps_i$'s,  
and iterate the entire process, reducing the error in
each step until we reach a local minimum.
 However, computing such an $S$ is more
 costly, and the resulting $S$ may no longer have the structured
 properties that we rely on to execute Step 2 efficiently. We could also
iterate only through Steps 2 and 3. But this requires computing
$R$ explicitly, which can be relatively inefficient. 
Taking a single pass through  the framework in
Figure~\ref{overview2} yields fast algorithms. 
As we show in Section~\ref{consistency},    
we can efficiently compute the (consistent) output $y$ of Step 3,
without materializing $R$.  
Our algorithm applies to cases where $Q$ consists of marginal queries.
This improves over the previous result
from~\cite{DWHL11}, which is described as only being applicable up to
mid-size data cubes before the time cost becomes too high. 
}

In the common case that $S$ is invertible,
our framework decreases the error for
the Fourier and wavelet approaches from prior work. 
Computing the optimal noise budgets here is very fast, so this
improvement comes with only a small time overhead: less than 1 second
in our experiments.

To summarize, our contributions are as follows:
\begin{itemize}
\item We propose a framework for minimizing the error of
  differentially private answers. It improves on the accuracy of
  existing strategies, at minimal computation cost.
\item We develop fast algorithms within this framework for
  marginal queries. Our algorithms compute consistent answers. 
In particular, when $Q$ is the set of all $k$-way marginals, we give
asymptotic bounds on the error of our mechanism; we are not aware of
any such analysis for the matrix mechanism.  
As a by-product, our analysis also improves the error bound for the
uniform noise case. 
\item We conduct an extensive experimental study on marginal query
  workloads and show that our framework reduces the error of
  existing strategies (including the Fourier strategy~\cite{BCDKST07} and the  
Cluster strategy~\cite{DWHL11}). 
\end{itemize}

\noindent
{\bf Organization.}
Section~\ref{defs} introduces the
necessary definitions for describing our framework. 
The optimization results required by Steps 2 and 3 are developed in
Section~\ref{framework}.  
In Section~\ref{consistency}, we describe novel results that allow us
to apply our framework to marginal queries in an efficient manner, 
and to compute consistent results. 
Our experimental study is presented in Section~\ref{sec:experiments}, and we
conclude in Section~\ref{concl}. 

\section{Definitions}\label{defs}

We begin by recalling the definition of differential privacy and some
fundamental mechanisms which satisfy this definition.

\begin{definition}[Differential privacy~\cite{DMNS06,DKMMN06}]
A randomized algorithm $\mathcal A$ satisfies $(\epsilon,
\delta)$-differential privacy if for all databases $D_1$ and $D_2$
differing in at most one element, and all measurable subsets $S
\subseteq \operatorname{Range}(\mathcal A)$,
$$\Pr[\mathcal A(D_1) \in S] \le e^{\epsilon} \cdot \Pr[\mathcal
  A(D_2) \in S] + \delta.$$
\end{definition}

We say that an algorithm satisfies $\epsilon$-differential privacy if
it satisfies $(\epsilon, 0)$-differential privacy.

%


\begin{definition}[$L^p$-sensitivity]
\label{def:lpsensitivity}
For $p \ge 1$ let the $L^p$-sensitivity $\Delta_p (f)$ of a function $f \colon D \rightarrow \mathbb R^q$ be defined as:
$$\Delta_p (f) = \max_{D_1, D_2} \|f(D_1) - f(D_2)\|_p,$$
for all $D_1$ and $D_2$ differing in at most one element.
Here, \mbox{$\|\cdot\|_p$} denotes the standard $L^p$ norm, i.e.,
$\|x\|_p = (\sum_{i = 1}^n |x_i|^{p})^{1/p}$ for $x \in \mathbb R^n$.
\end{definition}

We rely on the following two basic mechanisms to construct
differentially private algorithms:

\begin{theorem}[Laplace mechanism~\cite{DMNS06}]\label{thm:laplace-mechanism}
If $f$ is a function $f \colon D \rightarrow \mathbb R^q$, then
releasing $f$ with additive $q$-dimensional Laplace noise with variance $2
\left(\frac{\Delta_1(f)}{\epsilon}\right)^2$ in each component
satisfies $\epsilon$-differential privacy.
\end{theorem}

\begin{theorem}[Gaussian mechanism~\cite{DKMMN06, MM09}]\label{thm:gaussian-mechanism}
If $f$ is a function $f \colon D \rightarrow \mathbb R^q$, then
releasing $f$ with additive $q$-dimensional Gaussian noise with variance
$\left(2 \Delta^2_2 (f)\frac{\log(2 /
\delta)}{\epsilon^2}\right)$ in each component
satisfies $(\epsilon, \delta)$-differential privacy.
\end{theorem}

\noindent
{\bf Query workloads, consistency, strategy and recovery.}
As mentioned in Section~\ref{intro}, we represent the database as a vector
$x \in \mathbb R^N$ and the query workload as a matrix $Q \in \mathbb R^{q \times N}$: each row $Q_{i\cdot}$, $1\leq i\leq q,$
is a linear query over database $x.$
It is easy to see that the
sensitivity of $Q$ is $\Delta_p(Q) = \max_{j=1}^N \|Q_{\cdot j}\|_p,$
where $Q_{\cdot j}$ denotes the $j$th column of $Q$.\footnote{We assume that each individual contributes a weight of
  1 to some entry of $x$, in line with prior work.
Other cases can be handled by rescaling the sensitivity accordingly.}
One differentially private answer to $Q$ is a vector $y = Qx + \tau$, where $\tau \in \mathbb R^q$ is the noise vector drawn from
an appropriate (Laplace or Gaussian) distribution.
Our formal goal is to minimize the variance of a given linear
functional $a^T \cdot \Var(y)$ for some fixed vector $a \in \mathbb
R_+^q$, while guaranteeing  differential privacy.
For example, if $a = \overrightarrow{\mathbf 1}$
we minimize the sum of the variances of noise over all queries.
In particular, we study workloads $Q$ that consist of marginals
over $x$, such as the set of all $k$-way marginals, for some small integer $k$.

\begin{definition}\label{def-consist}
A noisy output $y = Qx + \tau$ is {\em consistent} if there exists at
least one vector $x^c$ such that $y=Qx^c.$
\end{definition}

We decompose a query workload $Q$ into a {\em strategy} matrix $S \in \mathbb R^{m \times N}$,
and a {\em recovery} matrix $R \in \mathbb R^{q \times m}$,  such that $Q = RS$. The query answer $y$ is then computed as $y=Rz$,
where $z=Sx + \nu$ is the noisy answer to $S$ (hence, $\tau = R\nu$).
In general, there are many possible ways to pick $R$ and $S$ given $Q$, and our goal will be to minimize the resulting $\Var(y).$

\section{Our Framework}\label{framework}
In this section we solve the optimization problems
required by Steps 2 and 3 of our framework from
Figure~\ref{overview2}.

\subsection{Optimal Noise Budgeting (Step 2)}\label{noisebudget}

A novel part of our scheme is a special purpose \textit{budgeting mechanism}:
For each row $S_{i\cdot}$ in the strategy $S$, we release
$z_i=S_{i\cdot} x + \nu_i$, where $\nu_i$ is drawn from a Laplace
distribution that depends on a value $\eps_i$.
We show how to choose the values $\eps_i$ optimally so that the overall method satisfies
$\eps$-differential privacy and the resulting noise is minimized.
We also design an approach based on {\em grouping} rows of the
strategy matrix $S$,
which allows us to compute the optimal $\eps_i$'s efficiently.

 \begin{proposition}\label{prp:nonuniform-budgeting}
Let $S$ be an $m\times N$ strategy matrix, and let
$\eps_1,\ldots,\eps_m$ be a set of $m$ non-negative values.
Define the noisy answer to $S$ to be an $m$-dimensional vector $z$ such that
$z_i=S_{i\cdot} x + \nu_i$, $1\leq i\leq m.$

(i) If $\nu_i$ is drawn from the Laplace distribution with variance $\frac{2}{\eps_i^2},$ then
$z$ satisfies $\alpha$-differential privacy, where $\alpha = 2\max_{j=1}^N (\sum_{i=1}^m |S_{ij}|\eps_i)$.

(ii) If $\nu_i$ is drawn from the Gaussian distribution with variance
$2\frac{\log (2/\delta)}{\eps_i^2},$ and $\alpha = 2\max_{j=1}^N
\sqrt{\sum_{i=1}^m S^2_{ij}\eps^2_i}$,
$z$ satisfies $(\alpha,\delta)$-differential privacy.
\end{proposition}

\begin{proof} 
\ifnum\final=0
We only show (i), the proof for (ii) is shown in~\ref{app:approx-dp-budgeting}.
\else
\fi
We decompose $S$ as $D^{-1}DS$ where $D$ is the diagonal matrix
$D = \diag(\eps_1, \ldots \eps_m)$.
We now consider the $L^p$ sensitivity of the function $f(x) = (DS)x$.
From Definition~\ref{def:lpsensitivity},
we have
\[ \Delta_p(f) \leq 2\max_{j=1}^{N} \| (DS)_{\cdot j}\|_p =
2\max_{j=1}^{N} (\sum_{i=1}^{m} |S_{ij}\eps_i|^p)^{1/p}
\]
Thus, adding noise with variance proportional to
$(2\frac{\Delta_p(f)}{\alpha})^2$ provides
$\alpha$-differential privacy (via Theorem~\ref{thm:laplace-mechanism}
with $p=1$) or $(\alpha,\delta)$-differential privacy (via
Theorem~\ref{thm:gaussian-mechanism} with $p=2$).
Finally, multiplying by $D^{-1}$ has the effect of rescaling the
variance in each component:
  the $i$th component now has variance proportional to
  $(\frac{\Delta_p(f)}{\alpha \eps_i})^2$.
Setting $\alpha = \Delta_p(f)$ for $p=1$ or $p=2$
and applying the correct scaling constants gives the claimed result.
%
\end{proof}

\ifnum\final=0
The proofs for (ii) and other results are omitted for brevity.
\fi

Recall that the output is computed as $y =Rz$.
Our goal is to choose values $\eps_i$ that
minimize the variance $a\transpose \Var(y) = a\transpose \Var(R\nu).$
We detail this for Laplace mechanism:
\eat{
\begin{align*}
a\transpose \cdot \Var(R\nu)  &
= \sum_{i = 1}^{q} a_i \cdot
  \Var\left(\sum_{j = 1}^m R_{ij} \cdot
  \Lap\bigg(\frac{1}{\eps_j}\bigg)\right) \\ &= 2
  \sum_{i = 1}^{q} a_i \sum_{j = 1}^m \frac{R^2_{ij}}{\eps^2_j} \\
&= 2 \sum_{i = 1}^m \frac{1}{\eps^2_i} \sum_{j = 1}^q a_j R^2_{ji}
\end{align*}
}
$$a\transpose \cdot \Var(R\nu) = 2 \sum_{i = 1}^{q} a_i \sum_{j = 1}^m \frac{R^2_{ij}}{\eps^2_j}
 = 2 \sum_{i = 1}^m \frac{1}{\eps^2_i} \sum_{j = 1}^q a_j R^2_{ji}.$$
Let $b_i$ $=$ $2\sum_{j = 1}^q a_j R^2_{ji}$.
By Proposition~\ref{prp:nonuniform-budgeting}, it follows that the optimal noise
budgeting $\{\eps_i\}$ is the solution to the following optimization problem:
\begin{align}
\text{Minimize: } & \textstyle
\sum_{i = 1}^m \frac{b_i}{\eps_i^2} \label{optpb}\\
\mbox{Subject to: } & \textstyle
\sum_{i=1}^m |S_{ij}|\eps_i\leq \eps,\ \ 1\leq j\leq N. \label{ineq:privacy}\\
& \epsilon_i \ge 0, \ \ 1 \le i \le m \label{ineq:eps}
\end{align}
Because all $b_i$'s are non-negative, the objective function is convex.
The body defined by the linear inequalities is also convex. The resulting problem can thus
 be solved using a convex optimization package that implements, e.g., interior point methods.
Such methods require time polynomial in $m$, $N,$ and the required accuracy of the solution~\cite{BV04}.
\eat{
Thus, any local minimum for the problem above is also a global minimum.
If we use the Kuhn-Tucker method to find a local minimum, then the corresponding generalized Lagrange function is:
$$\Lambda(\lambda_1,\ldots,\lambda_N, \eps_1, \dots, \eps_m) =
 \left(\sum\limits_{i = 1}^m \frac{b_i}{\eps^2_i} \right) +
 \sum_{j=1}^N\lambda_j\left(\sum_{i = 1}^m |S_{ij}|\eps_i - \eps\right).$$
Taking partial derivatives $\frac{\partial}{\partial \eps_i}$ we
have a condition for the extremum:
 $$ \frac{2 b_i}{\eps^3_i} = \sum_{j=1}^N\lambda_j|S_{ij}| \qquad \forall i \in [m],$$
so we obtain $\eps_i$ $=$
  $\Big(\frac{2b_i}{\sum_{j=1}^N\lambda_j|S_{ij}|}\Big)^{1/3}.$
We substitute these $\eps_i$'s into the other Kuhn-Tucker conditions:
$$\lambda_j\left(\sum_{i = 1}^m |S_{ij}|\eps_i - \eps\right) = 0,$$
to obtain a system of $N$ non-linear equations in $N$ variables $\lambda_j.$ The system
admits a solution with $\lambda_j\geq 0,\ \forall j$, and such that $\sum_{i=1}^m |S_{ij}|\eps_i\leq \eps.$
One can approximate the solution $\{\lambda_j\}$ to this system via the iterative Newton method, then compute the
corresponding $\{\eps_i\}.$
\grinote{Sounds a bit vague.}
This iterative approach requires inverting an $N\times N$ Jacobian matrix during
each iteration.\grinote{An alternative approach is to use an SDP, as in Matrix Mechanism, which also results in an iterative process for $L_1$-sensitivity, but has provably polynomial running time for $L_2$.}
}

\medskip
\noindent
{\bf Efficient Solution via Grouping.}
Convex optimization solvers may require a large number of iterations and be too inefficient
for databases of moderate dimensionality.
However, for most of the frequently used strategy matrices, the
optimization problem can be significantly simplified, if we partition
the rows of the strategy matrix $S$ into {\em groups}, and define the
corresponding values $\eps_i$ to be the same for all rows in a group.
We show that the groups can be chosen in such way that all
conditions $\sum_{i=1}^m |S_{ij}|\eps_i\leq \eps$ become {\em identical}
once we set the $\eps_i$'s to be equal in each group, which
leads to a closed form solution.
This approach was implicitly used in~\cite{CPSSY12}.
We show that this concept
can be applied to a larger class of strategy
matrices.
The optimal solution for the simplified problem is a feasible solution
for the general problem.
If recovery matrix
$R$ satisfies a
certain property (as is the case for {\em all} matrices we consider),
then the optimal solution
for the simplified problem is also guaranteed optimal for the general
case. In particular, we
find optimal noise budgets for
strategy/recovery methods such as Fourier~\cite{BCDKST07} and
clustering~\cite{DWHL11}.

\begin{definition}\label{grouping}
Let $S$ be an $m\times N$ strategy matrix. We say that $S$ satisfies
the {\em grouping property} if there exists a grouping function over its rows
$G: [m]\rightarrow [g]$, $g \leq m$, such that the following two
conditions are satisfied:

\noindent
--- row-wise disjointness:
for any two rows $i_1,i_2$ of $S$ with $G(i_1)=G(i_2)$ and for any
column $j$, $S_{i_1j} S_{i_2j} = 0;$

\noindent
--- bounded column norm: for any group $r$, and for any two columns
$j_1,j_2$, we have
$\max_{i: G(i)=r} |S_{ij_1}| = \max_{i: G(i)=r} |S_{ij_2}| = C_r.$

The minimum $g$ for which $S$ has a grouping function $G$ is called the {\em grouping number} of $S$.
\end{definition}

Together, the two conditions in Definition~\ref{grouping} imply that
any column of $S$ contains at most one non-zero value from each group,
and that that value is the same (within a group) for all columns.
Hence, not every $S$ can meet this definition: while we could put every
row in a singleton group, we also then require that the magnitude
of all non-zero entries in the row are identical.
Nevertheless, as we show below,
many commonly used matrices are groupable.

\smallskip
\noindent
{\bf Example.}
Matrix $S$ in Figure~\ref{toyex}(c) has grouping number $g=1$: each
column has exactly one entry equal  to 1, so $C_1 = 1$.
On the other hand, if $S=Q$ is the matrix in Figure~\ref{toyex}(b),
the grouping number is 2: we define one group containing
the first two rows, and another containing the last four rows.
We have $C_1=C_2=1.$
Note that, e.g., the first and third rows cannot be grouped
together, since $Q_{11}Q_{31}=1\neq 0.$
We now apply this definition to the other strategy matrices proposed:

\smallskip
\noindent
{\em Base counts.}
As noted in the introduction, directly materializing the noisy version
of $x$ is equivalent to $S=I.$ In this case, all rows form a single group; hence, $g=1$ and
$C_1=1$.

\smallskip
\noindent
{\em Collections of marginals.}
When $S$ is a set of marginals, all rows that compute
the cells in the same marginal can be grouped together, as in the
above example.
Hence, the number of groups $g$ is the number of marginals computed; and $C_r=1$ for
each group $r$.

\smallskip
\noindent
{\em Hierarchical structures.}
When $S$ represents a hierarchy over $x$, all rows that compute the counts at the same level
in the hierarchy form one group. Hence, the grouping number $g$ is
the depth of the hierarchy and all $C_r$ values are 1.
Specifically, when $S$ represents a binary tree over $x$,
the grouping number is $g = \lceil\log_2 N\rceil.$
The same essentially holds for the one-dimensional Haar wavelet (here,
$g=\lceil \log_2 N\rceil +1$).
For higher dimensional wavelets, the grouping number grows
exponentially with the dimension of the wavelet transform.

\smallskip
\noindent
{\em Fourier transform.}
The Fourier transform (discussed in more detail in Section~\ref{sec:fourierdef})
is dense: every entry is non-zero and has absolute value $2^{-d/2}$. In this case, each row forms
its own group, the grouping number is $N$, and $C_r = 2^{-d/2}$ for any group $r$.

\smallskip
\noindent
{\em Sparse random projections.}
Sketches are sparse random projections that partition the data $x$
into buckets, repeated $t$ times~\cite{CPST12}.
All entries in the sketch matrix $S$ are $\{-1,0,+1\}$. In this case, all rows that define one
particular partition of the data form one group, so $g=t$ and $C_r=1$.

\smallskip
\noindent
{\em Arbitrary strategies $S$.}
If $S$ is groupable, we can greedily find a grouping as follows: start a
group with an arbitrary row, and try to add each remaining row to existing
groups; if a row cannot be added to an existing group, a new group is
created for it.
While this may not result in a minimum $g$, any grouping suffices for
our purposes.
We do not discuss the greedy approach further,
since all the strategies we study can be grouped directly as discussed above.

\begin{definition}\label{Rgrouping}
Let $S$ be an $m\times N$ strategy matrix with grouping function $G$. Let $R$ be a corresponding
$q\times m$ recovery matrix. We say that $R$ is {\em consistent with $G$} if for any
rows $i_1,i_2$ of $S$ with $G(i_1)=G(i_2)$, we have $b_{i_1} = b_{i_2}$ (where $b_i$ $=$ $2\sum_{j = 1}^q a_j R^2_{ji}$ are as in objective function~\eqref{optpb}).
\end{definition}

When $Q$ is a set of marginals and $a = \overrightarrow{\mathbf 1},$ it is easy to verify that
$R$ is consistent with the optimal grouping of $S$, for all the choices of $S$ considered in prior work: $S=I$, $S=Q$, $S=$Fourier transform, and $S=$ strategy marginals computed by clustering~\cite{DWHL11} (here, $R$ aggregates cells of the centroid marginal to compute each
of the marginals assigned to a cluster).

\smallskip
The next result follows directly from the properties of the grouping function.

\begin{lemma}\label{groupinglemma}
Let $S$ be a strategy matrix with grouping function $G.$ There is a feasible solution to the
optimization problem~\eqref{optpb} -- \eqref{ineq:eps} such that
for each group $r$ and for all pairs of rows $i_1,i_2$ with $G(i_1) =
G(i_2) = r$, we have $\eps_{i_1} = \eps_{i_2}.$
Moreover, all privacy conditions \eqref{ineq:privacy} are equivalent,
and can be satisfied with equality.

If $R$ is consistent with $G$, then the above solution is optimal for the problem defined
by~\eqref{optpb} -- \eqref{ineq:eps}.
\end{lemma}

\begin{proof}
Let $\eta = \eta_1,\ldots,\eta_g$ be the noise budgets corresponding to the
$g$ groups of $S$; i.e., all $\eps$ values for the rows
in group 1 are equal to $\eta_1$, etc. Because of the grouping property, each condition~\eqref{ineq:privacy}
becomes $\sum_{i=1}^g C_i\eta_i \leq \eps,$ where $C_i$ is the value defined by the bounded column norm for
the group $i$ (recall Definition~\ref{grouping}). Since the objective function is a minimization, we can make this inequality an equality. Clearly,
$\{\eta_i\}_i$ are a feasible solution for~\eqref{optpb} -- \eqref{ineq:eps}.

If $R$ is consistent with $G$, we can change any optimal solution
of~\eqref{optpb} -- \eqref{ineq:eps} into a solution in which all $\eps$ values in a group are equal, without increasing the objective function.
We omit a formal proof here.
\end{proof}

Thus, when $S$ has grouping function $G$, we can write a simpler optimization problem for noise budgeting:
\begin{align}
\text{Minimize: } & \textstyle \sum_{i = 1}^g \frac{\sum_{r: G(r)=i} b_r}{\eta_i^2} \label{soptpb} \\
\mbox{Subject to: } & \textstyle \sum_{i=1}^g C_i\eta_i = \eps. \label{eq:privacy}\\
& \eta_i \ge 0, \ \ 1 \le i \le m \label{ineq:eta}
\end{align}
Since there is now just a single constraint on the $\eta_i$s,
we can solve this via a simple Lagrange multiplier method.
The corresponding Lagrange function is:
$$\textstyle
\Lambda(\lambda, \eta) =
 \big(\sum\limits_{i = 1}^g \frac{\sum_{r: G(r)=i} b_r}{\eta^2_i} \big) +
 \lambda\big(\sum\limits_{i = 1}^g C_i\eta_i - \eps\big).$$

\noindent
Setting the partial derivatives $\frac{\partial}{\partial \eta_i}$ to
zero, we obtain $\eta_i = \big(\frac{2}{\lambda C_i} \sum_{r: G(r)=i} b_r
\big)^{1/3}.$
By the privacy constraint (\ref{eq:privacy}), $\sum\limits_{i = 1}^g \left(\frac{2 C_i^2}{\lambda}\sum_{r: G(r)=i} b_r\right)^{1/3} = \eps$ and thus:
\[ \textstyle
  \lambda = \frac{2}{\eps^3}
   \big(\sum\limits_{i = 1}^g \big(C_i^2\sum\limits_{r: G(r)=i} b_r\big)^{1/3}\big)^3
%
\]
\begin{corollary}~\label{cor:noisebudget-variance}
In the case when all values $C_i$ are equal to the same value $C$ the
optimum value of the objective function is equal to
$\frac{C^2}{\epsilon^2} \left(\sum_{i = 1}^g s_i^{1/3}\right)^3$,
where $s_i = \sum_{r \colon G(r) = i} b_r$.
For $(\epsilon, \delta)$-differential privacy the corresponding value
of the objective function is equal to
$\frac{2 C^2 \log(2/\delta)}{\epsilon^2} \left(\sum_{i = 1}^g
\sqrt{s_i}\right)^2$.
\end{corollary}

Lemma~\ref{groupinglemma} implies the following.

\begin{theorem}\label{thm:noisebudgets}
 Let $S$ be a strategy matrix with grouping function $G$, and $R$ be a corresponding recovery matrix
consistent with $G$. Then the solution to the optimization
problem~\eqref{soptpb} -- \eqref{ineq:eta} is the optimal
noise budgeting for $S$ and $R$.
\end{theorem}

As discussed above, when $Q$ is a set of marginals, all the strategy/recovery matrices proposed in prior work fit
the conditions of Theorem~\ref{thm:noisebudgets}, and thus their accuracy can be improved via optimal noise budgeting.
Observe that the optimization problem~\eqref{soptpb} --
\eqref{ineq:eta} can be solved reasonably fast: given $R$,
$S$ and a grouping of $S$, we can derive the vector of $b_i$ values in time linear in the size of $R$, i.e., in $O(qm)$.
For particular $S$ (e.g. the Fourier matrix), the cost can
be even lower, due to the symmetric structure of $S$ and $R$. Finally, all $\eps_i$ values, as well as $\Var(y)$,
can be computed in $O(m)$ time.

\subsection{Optimal Recovery Matrix (Step 3)}\label{ols}

Given $S \in \mathbb R^{m \times N}$ and the noise parameters $\eps_i$, we wish to compute a matrix
$R \in \mathbb R^{q \times m}$ such that $Q = RS$ and $\Var(y)$ is minimized. Recall that $y=Rz=R(Sx+\nu)$, where
$\nu_i$ is drawn from the Laplace distribution of variance $\frac{2}{\eps_i^2}$, as in
Section~\ref{noisebudget} (the case of Gaussian distribution is similar).
As we show below, the resulting $y$ will also be consistent.

We derive $R$ via least squares statistical estimators. More precisely, given $z=Sx+\nu$,
we first compute an estimate $\hat{x}$ of $x$ which is linear in $z$
and has minimum variance. The vector $\hat{x}$ is called the
optimal (generalized) least squares solution. As we show below, $\hat{x}=Gz$ for some matrix $G$. We
then define $R=QG$ and $y=Rz=Q\hat{x}$. A similar approach was used in~\cite{LHRMG10} for the case
of uniform noise. We extend the computation to the case of non-uniform noise budgets $\eps_i$.

Let $\Sigma$ be the covariance matrix of $z$: $\Sigma = \Cov(z) = \diag(\frac{2}{\eps_i^2}).$
Define $U = \Sigma^{-1/2} S$; hence, $\operatorname{rank}(U)= \operatorname{rank}(S).$
For simplicity, we assume that $\operatorname{rank}(S)=N$. The same ideas as in~\cite[Section 3.3]{LHRMG10}
can be used to handle the case $\operatorname{rank}(S) < N$; see also~\cite{Rao:65} for further details.
Then the LS solution is computed as
\[\hat{x}   = (U\transpose U)^{-1}U\transpose \Sigma^{-1/2}z.\]
Since $\Sigma$ is diagonal, $\Sigma=\Sigma\transpose$. We obtain
\begin{align*}
(U\transpose U)^{-1}U\transpose & =
(S\transpose\Sigma^{-1/2}\Sigma^{-1/2}S)^{-1}S\transpose \Sigma^{-1/2} \\ & =
(S\transpose \Sigma^{-1}S)^{-1}S\transpose \Sigma^{-1/2}.\end{align*}
\[
\text{Thus, }
\hat{x}  =  (S\transpose \Sigma^{-1}S)^{-1}S\transpose \Sigma^{-1} z.\]

Let $G=(S\transpose \Sigma^{-1}S)^{-1}S\transpose  \Sigma^{-1}$. We define $R=QG$, i.e.,
\begin{equation}\label{nonunifR}
R = Q (S\transpose \Sigma^{-1}S)^{-1}S\transpose  \Sigma^{-1}.
\end{equation}

Note that $y = Rz = Q\hat{x}$ is consistent, as per Definition~\ref{def-consist}
(with $x^c=\hat{x}$). By a well-known result from linear statistic estimation~\cite{Rao:65}, the following holds:

\begin{lemma}\label{lem:budgeted-recovery}
Matrix $R$ computed as in~\eqref{nonunifR} minimizes $a\transpose \Var(y)$ (where $y = Rz$). Moreover,
$y$ is consistent and unbiased, i.e., $E[y]=Qx$.
\end{lemma}

\begin{observation}
If $S$ is an orthonormal basis (as with wavelets, Fourier and
identity strategies), we have $S\transpose=S^{-1}$. This implies $G=S^{-1}=S\transpose$, so
$R=QS\transpose$.
\end{observation}

The cost of finding $R$ as above is relatively high, due to the need to
perform matrix inversion.
While the diagonal matrix $\Sigma$ is trivial to invert, since
$\Sigma^{-1}_{ii} = (\Sigma_{ii})^{-1}$, the matrix
$S\transpose \Sigma^{-1} S$ is generally dense, so is more costly to
invert.

\subsection{Fast Consistency}\label{subsec:fastconsist}

The vector $y=Rz=R(Sx+\nu)$ computed via the optimal recovery matrix $R$ in Step 3 (Section~\ref{ols})
 has two important properties: (i) $y$ is consistent (Definition~\ref{def-consist}); and (ii) $\Var(y_i)$ is minimized for each $1\leq i\leq q$.
Since $\mathbb E[y_i] = Q_{i\cdot}x$ we have $\Var(y_i)=\mathbb E[(y_i-Q_{i\cdot}x)^2]$. Thus, $y$ achieves minimum error for each
query in expectation. As observed in~\cite{BCDKST07, DWHL11} for practical applications it may be necessary to return a vector $y^1$
which is consistent and minimizes a different error measure, e.g., we may wish to minimize $\|y^1-Qx\|_p.$
For example, $p=1$ implies that $y^1$ minimizes the average error and $p = \infty$ minimizes maximum error.

In this section we show how to efficiently compute another recovery matrix $R^1$ such that
$y^1=R^1z$ is consistent and $\|y^1-Qx\|_p$ is small. This approach
is particularly useful when the query matrix $Q\in \mathbb R^{q \times N}$
has $q\ll N$. As we show below, we significantly improve the running times of the approaches used in~\cite{BCDKST07,DWHL11} for this case.
The approach in~\cite{BCDKST07,DWHL11}, translated in our strategy/recovery
framework, is described below.

Start by defining a recovery matrix $R^0$ such that $Q=R^0 S$ and
$y^0=R^0(Sx+\nu)$ has bounded error $\|y^0 - Qx\|_p\leq t$. Usually, $R^0$ is the recovery matrix from Step 2
of our framework. For example, the matrix $R^0$ in~\cite{DWHL11} is implied by the clustering function
over marginals, which heuristically minimizes some
$L^p$-error of the noisy answers. Next, compute a consistent answer $y^1$ that minimizes $\|y^1 - y^0\|_p$.
Recall that $y^1$ is consistent if there exists $x^c$ such that $y^1=Qx^c$. Hence, for $p=1$ or $p=\infty$,
$y^1$ can be computed via a linear program (LP) with variables corresponding to the entries of
$x^c$, and consistency conditions expressed as linear constraints. Other requirements can also be imposed
on $x^c$, e.g., integrality or non-negativity. For $p=2$, $y^1$ is the solution to a least
squares problem (LS).

However, such an LP, resp. LS, uses at least $N$ variables corresponding to the entries in $x^c$.
When $N$ is large (as is usually the case), this leads to large
linear programs. This issue was reported as a bottleneck in the experimental evaluation of~\cite{DWHL11}.
We now propose a different LP, resp. LS, formulation for the consistency problem,
which requires at most $q$ variables (recall that $q$ is the number of queries in the workload $Q$). This leads to
large improvements in running time when $q\ll N$.

First, note that $\operatorname{rank}(Q)=q$ implies that any answer $y\in\mathbb R^q$ is consistent. This is because the linear system $Qx^c = y$ admits the solution $x^c = Q\transpose (QQ\transpose)^{-1}y$ ($\operatorname{rank}(Q)=q$ implies that $QQ\transpose$ is invertible). In particular, $y^1=y^0$ is consistent and minimizes
$\|y^1 - y^0\|_p$ for any $p$.

Assume that $\operatorname{rank}(Q) = q' < q$. We pick $q'$ linearly independent rows of $Q$,
 denoted as $Q' \in \mathbb R^{q' \times N}$, and use them to
 decompose $Q$ as $Q = C Q'$ for some matrix
 $C \in \mathbb R^{q \times q'}$.
Because $Q'$ has linearly independent rows, the above argument implies that, for any $y \in \mathbb
R^{q'}$, the linear system $Q'x^c = y$ has a solution. Hence, any answer $y$ is consistent for the queries $Q'$.
Then $y^1 = Cy$ is consistent for all queries $Q$: $y^1 = Cy = C Q'x^c = Qx^c.$
We find $y$ that minimizes $\|C y - y^0\|_p$ and return $Cy$:
For $p = 1$ and $p = \infty$, $y$ is the solution to an LP; for $p=2$, $y$ is the solution to a least squares problem.
In all cases, the number of variables is $q' < q\ll N$. As observed in~\cite{BCDKST07}, the utility guarantee
follows by the triangle inequality. If $\|Qx - y^0\|_p \le t$, then
$$\|y^1-y^0\|_p = \min_{y \in \mathbb R^{q'}}\|Cy - y^0\|_p \le \|C Q' x - y^0\|_p \le t.$$
Thus, the additional $L_p$-error introduced by consistency is at most the $L_p$-error of
the original noisy answer, i.e., the error at most doubles.

When $Q$ is a set of marginals, we can formulate the LP, resp. LS, without explicitly computing $\operatorname{rank}(Q)$ or
finding a collection of
linearly independent rows $Q'$. Rather, we use the Fourier coefficients of the marginals. The discussion is deferred
to Section~\ref{sec:consistency}.

\section{Consistent Marginals via Fourier Strategies}
\label{consistency}

In this section, we focus on the case when all queries $Q$ correspond
to marginals.
Here, we show that the choice of $S$ as an appropriate Fourier matrix
gives strong guarantees on the variance, as well as
providing consistent query answers.

\eat{
\begin{table}[t]
\centering
\begin{tabular}{|l|l|l|}
  \hline
  Strategy
   & $\eps$-privacy & $(\eps, \delta)$-privacy \\
  \hline
    Marginals & $O\left(\frac{2^{\|\beta\|} \ell}{\eps}\right)$ 
    \cite{BCDKST07}& $O\left(\frac{2^{\|\beta\|} \sqrt{\ell \log(1 / \delta)}}{\eps}\right)$
    \cite{BCDKST07, KRSU10}\\
  \hline
    Fourier &
$O\left(\frac{|\F|\sqrt{2^{\|\beta\|}}}{\eps}\right)$
{\small Appendix~\ref{app:marg}}
& $O\left(\frac{\sqrt{|\F| 2^{\|\beta\|} \log (1 / \delta)}}{\eps}\right)$
    \cite{BCDKST07, KRSU10} \\
  \hline
\end{tabular}
\caption{
Expected noise per marginal: $\mathbb{E}\left[\|C^\beta x - \tilde C^{\beta}\|_1\right]$.
Here $\ell$ is the total number of released marginals and $|\F|$ is the total number of Fourier
coefficients required to compute these marginals.
}\label{tab:noise-summary}
\end{table}
}

\subsection{Marginals and Fourier analysis}
\label{sec:fourierdef}


In this section we assume that all $d$ attributes in the database table are binary; for simplicity, let the domain of each attribute be $\{0,1\}$. We emphasize that this assumption is
without loss of generality: an attribute which has $|A|$
distinct values can be mapped to $\lceil \log |A|
\rceil$ binary attributes (and we do so in our experimental
study). However, we present our results with binary attributes to
avoid overcomplicating the notation.
Consequently, there are $N = 2^d$ entries in the database vector $x$, where
each entry is indexed by some $\alpha \in \{0,1\}^d$, and $x_\alpha$ is the number of entries
in the database with attributes $\alpha$; recall the example in Figure~\ref{toyex}(a).

There are also $2^d$ possible marginals (a.k.a. subcubes of the data cube) of interest, corresponding
to aggregations along a subset of dimensions. For any $\alpha \in \{0,1\}^d$, let $C^\alpha$ denote the
marginal over non-zero attributes in $\alpha$, and let $\|\alpha\|$ denote the number of non-zero
entries in $\alpha$, i.e., the dimensionality of the marginal. Note that here $\alpha$ is the
bit-vector indicator for the attributes in the marginal. We will consistently use it as a superscript in such cases, and as a subscript when it indexes an entry in a vector.


We use the following notations, as in~\cite{BCDKST07}: For any pair of $\alpha, \beta \in
\{0,1\}^d$ we write $\alpha \wedge \beta$ for the bit-wise
intersection of the pair, i.e. $(\alpha \wedge \beta)_i = \alpha_i
\wedge \beta_i$.
The inner-product in this
space, $\langle \alpha, \beta \rangle$, can also be expressed via the
intersection operator:
$\langle \alpha, \beta \rangle = \|\alpha \wedge \beta\| $.
We say that $\alpha$ is {\em dominated by} $\beta$, denoted
 $\alpha \preceq \beta$, if $\alpha \wedge \beta = \alpha$.

The computation of a marginal $C^\alpha$ over the input can be thought
of as a linear operator $C^\alpha \colon \mathbb R^{2^d} \rightarrow \mathbb R^{2^{\|\alpha\|}}$
mapping the full-dimensional contingency table to the marginal over non-zero
attributes in $\alpha$, by adding relevant entries over the attributes not in $\alpha$. More precisely,
for each $\beta \preceq \alpha$, the cell $\beta$ in the marginal $C^\alpha$, denoted $(C^\alpha x)_\beta$,
sums the entries in the contingency table $x$ whose attributes in $\alpha$ are set to values
specified by $\beta$:
$(C^\alpha x)_\beta = \sum_{\gamma \colon \gamma \wedge \alpha = \beta} x_\gamma.$

\noindent
{\bf Example.} Let $x$ be the vector in Figure~\ref{toyex}(a). Assume we want to compute the marginal
$C^\alpha = C^{110}$, i.e., the marginal over attributes $A$ and $B$.
Then the value in the cell $(A=0,B=0)$
is denoted by $(C^{110}x)_{000}$ (i.e., $\beta=000$).
The value in the cell $(A=0,B=1)$ is denoted
by $(C^{110}x)_{010}$.
Note that $000 \preceq 110$ and $010 \preceq 110$. On the other hand,
$001\not\preceq 110$, so there is no cell $(C^{110}x)_{001}$ in the
marginal over $A,B$.
So, while the cell index $\beta$ is $d$-dimensional,
only the $\|\alpha\|$ bits corresponding to non-zeros in $\alpha$
vary---the rest are held at 0.
Hence, there
are only $2^{\|\alpha\|}$ cell indexes in the marginal $C^\alpha x$.
In this example, there are only 4 cells in $C^{110}x$.
By the above formula, $(C^{110}x)_{000}$ $=$ $x_{000} + x_{001}$ $=$ $3$ and
$(C^{110}x)_{010}$ $=$ $x_{010} + x_{011}$ $=$ $1.$
\qed

The set of all marginals $C^\alpha$ with $\|\alpha\| = k$
is referred to as the set of all $k$-way marginals.
They are commonly used to visualize the low-rank dependencies
between attributes, to build efficient classifiers from the data,
and so on.

We use the Hadamard transform, which is the
$2^d$-dimensional discrete Fourier transform
over the Boolean hypercube $\{0,1\}^d$.
This allows us to represent
any marginal as a summation over relevant Fourier coefficients.
The advantage is that the number of coefficients needed for each
marginal is just the number of entries in the marginal.
The Fourier basis vectors $f^\alpha$ for $\alpha \in \{0,1\}^d$ have components
$f^\alpha_\beta = 2^{-d/2} (-1)^{\langle \alpha, \beta \rangle}.$
The vectors $f^\alpha$ form an orthonormal basis in $\mathbb R^{2^d}$.
We will use the following properties of Fourier basis vectors and
marginal operators in the Fourier basis (proofs can be found
in~\cite{BCDKST07}):

\begin{theorem}\label{thm:fourier-basics}
For all $\alpha, \beta \in \{0,1\}^d$ we have:
  \begin{enumerate}
    \item $(C^\alpha f^\beta)_\gamma = \sum\limits_{\eta \colon \eta \wedge \alpha = \gamma} f^\beta_\eta = \sum\limits_{\eta \colon \eta \wedge \alpha = \gamma} (-1)^{\langle \beta, \eta \rangle}/2^{d/2}$.
    \item $C^\alpha x = \sum\limits_{\beta \preceq \alpha} \langle f^\beta, x \rangle C^\alpha f^\beta$
  \end{enumerate}
\end{theorem}

\begin{table*}[t]
\centering
\begin{tabular}{|l|l|l|}
\hline
 Strategy  & $\eps$-privacy & $(\eps, \delta)$-privacy \\

  \hline
    Base counts & $O(\frac1\eps{2^{(d + k)/2}})$
     \cite{DMNS06} & $O(\frac1\eps {2^{(d + k)/2}
       \sqrt{\log(1/\delta)}})$
     \cite{DKMMN06}\\
  \hline

    Marginals & $O(\frac1\eps {2^{k} \binom{d}{k}})$
     \cite{BCDKST07} & $O(\frac1\eps{2^{k} \sqrt{\binom{d}{k} \log(1 / \delta)}})$
     \cite{BCDKST07}\\
  \hline

    Fourier coefficients (uniform noise)  & $O(\frac1\eps{k
      \binom{d}{k} \sqrt{2^{k}}})$
\ifnum\final=0
    Theorem~\ref{thm:improved-fourier-noise}
\fi
& $O(\frac1\eps{\sqrt{k 2^k \binom{d}{k} \log (1 / \delta)}})$
    \cite{BCDKST07}\\
  \hline
    Fourier coefficients (non-uniform noise)
& $O(\frac1\eps{k \sqrt{\binom{d}{k}\binom{d + k}{k}}})$ Lemma~\ref{lem:budgeting-k-way-marginals} & $O(\frac1\eps{\sqrt{k \binom{d + k}{k} \log (1 / \delta)}})$
Lemma~\ref{lem:budgeting-k-way-marginals}\\

  \hline\hline

  Lower bound & $\tilde \Omega(\frac1\eps{\sqrt{\binom{d}{k}}})$ \cite{KRSU10} & $\tilde \Omega(\frac1\eps{\sqrt{\binom{d}{k}} (1 - \delta/\eps)})$ \cite{KRSU10}\\
  \hline
\end{tabular}
\caption{
Releasing all $k$-way marginals for $k < d / 2$. Expected noise per marginal: $\mathbb{E}\left[\|C^\beta x - \tilde C^{\beta}\|_1\right]$.
The total number of released marginals is $\binom{d}{k}$ and the total number of Fourier coefficients required to compute these marginals is $\sum_{i = 0}^{k} \binom{d}{k} \le k \cdot \binom{d}{k}$.
}\label{tab:noise-k-way-summary}\end{table*}

\subsection{Bounds for marginals}
\label{marginals}

The use of a Fourier strategy matrix was studied in~\cite{BCDKST07},
under a uniform error budget.
Here, we show that using a non-uniform budgeting can provide
asymptotically improved results.
We study the case when the query set $Q$ corresponds to a
collection of $\ell$ marginals $C^{\alpha_1}, \dots, C^{\alpha_\ell}$.
For a given marginal $C^{\alpha_i}$ the accuracy bounds will be
parametrized by its dimensionality $\|\alpha_i\|$, the total number of marginals $\ell$ and the total number
of Fourier coefficients corresponding to the collection of marginals,
denoted as $|\F|$. Theorem~\ref{thm:fourier-basics}(2) implies that $|\F|$ $=$ $|\cup_i\{ \beta :
\beta
\preceq \alpha_i \}|$.
If the random variable corresponding to the differentially private value of a
marginal $C^\alpha x$ is denoted as $\tilde C^\alpha x$, then we state a
bound on the expected absolute error,
$\mathbb E\left[ \|C^{\alpha} x - \tilde C^{\alpha}x\|_1\right]$
to simplify presentation and comparison with prior work.
All our bounds can also be stated in terms of the variance
$\Var(\tilde C^{\alpha}x)$, or as high-probability bounds.

The asymptotic bounds on error are easier to interpret in the
important special case of the set of all $k$-way marginals.
In this case because of the symmetry of the query workload, the
expected error in all marginals is the same.
Table~\ref{tab:noise-k-way-summary} summarizes bounds on error in this
case together with the unconditional lower bounds for all
differentially private algorithms from~\cite{KRSU10}.
While in the case of $(\eps,\delta)$-differential privacy our upper
bounds are almost tight with the lower bounds from~\cite{KRSU10}, for
$\eps$-differential privacy the gap is still quite significant and
remains a challenging open problem.

Our next
\ifnum\final=0
 lemma (proof in Appendix~\ref{app:k-way})
\else
result
\fi
gives bounds on expected error of the Fourier
strategy  with {\em non-uniform} noise.
\ifnum\final=1
We omit the proof for lack of space.
\fi

\begin{lemma}~\label{lem:budgeting-k-way-marginals}
For a query workload consisting of all $k$-way marginals over data $x \in \mathbb R^{2^d}$ the bounds on the
expected error of the Fourier strategy mechanism with non-uniform noise are given as follows:\\
1. For $\eps$-differential privacy the expected noise per marginal is
$O\big(\frac1\eps \cdot k \sqrt{\binom{d}{k} \binom{d+k}{k}}\big)$.
\\
2. For $(\eps, \delta)$-differential privacy the expected noise per
marginal is $O\big(\frac1\eps \cdot {\sqrt{k \binom{d + k}{k}
    \log(1/\delta)}}\big)$.
\end{lemma}

%
%
These bounds are summarized in~\tabref{noise-k-way-summary}, along with those that follow from other approaches.%
\ifnum\final=1
\footnote{We derive a tighter bound for the Fourier stategy under
  uniform noise than in~\cite{BCDKST07}, but details are omitted for
  space reasons.}
Using the Fourier transform as the strategy matrix in our framework
improves substantially over  prior bounds:
  strategies based on $S=I$ or $S=Q$ incur factors exponential in $d$
  (so linear in $N$, the size of the contingency table).
In contrast, our bounds depend only on $k$, and factors in $d$ that
are polynomial for constant $k$ (the region of most interest).
That is, the noise depends on $\log N$.
\else
We note in passing that we can provide a tighter analysis of the noise for
the Fourier strategy under uniform noise than in~\cite{BCDKST07},
by a factor of $O(\sqrt{2^k})$---details are in
Appendix~\ref{app:marg}.
\fi

\medskip
\noindent
{\bf Time cost comparison.}
To directly compute a single $k$-way marginal over $d$-dimensional
data takes time $O(2^d)$, and so computing all $k$-way marginals takes
$O(d^k 2^d)$ naively.
Computing the Fourier transform of the data takes time $O(d 2^d)$,
and deriving the $k$-way marginals from this takes time $O(4^k)$
per marginal, i.e., $O(d 2^d + 4^{k} d^k)$ for all marginals~\cite{BCDKST07}.
We compare the cost of different strategies to these costs.
Materializing noisy counts ($S=I$) and aggregating them
to obtain the $k$-way marginals also takes time $O(d^k 2^d)$, as does
materializing the marginals and then adding noise ($S=Q$).

The clustering method proposed in~\cite{DWHL11} is more expensive, due
to a search over the space of possible marginals to output.
The cost is $O(d^k k \min(2^d d^k,3^d))$: clearly as the dimensionality
grows, this rapidly becomes infeasible.
However, across all strategies, the step of choosing the non-uniform
error budget is dominated by the other costs, and so does not alter
that asymptotic cost.

\subsection{Consistency via Fourier coefficients}\label{sec:consistency}

In Section~\ref{subsec:fastconsist} we discussed a general approach for computing consistent answers
for a query workload $Q\in\mathbb R^{q \times N}$ with $\operatorname{rank}(Q) < q$. The approach required explicitly
finding a set of $\operatorname{rank}(Q)$ linearly independent rows in $Q$ and decomposing $Q$.
We now show that when $Q$ is a set of marginals we can compute consistent answers without such expensive
steps. Instead, we ensure consistency by writing an LP that uses the Fourier coefficients corresponding to
the marginals in $Q$. Let $Q$ consist of $\ell$ marginals
  $C^{\alpha_1}, \dots, C^{\alpha_\ell}$.
We introduce variables for the Fourier coefficients corresponding to
these marginals, denoted as $\mathcal{F} = \{\hat f^{\beta} | \exists
i \colon \beta \preceq \alpha_i \}$. To simplify notation, we rename them as
$\mathcal F = \{\hat f_{1},\dots, \hat f_{m}\}$, where $|\mathcal F| = m$.
Marginals $C^{\alpha_1}, \dots, C^{\alpha_\ell}$ can be computed from
$\mathcal F$, using formulas from Theorem~\ref{thm:fourier-basics}:
$$(C^{\alpha_i})_\gamma = \textstyle{\sum_{\alpha \preceq \alpha_i}} \hat f^{\alpha} \left( C^{\alpha_i} f^\alpha \right)_\gamma,$$
for all $i \le \ell$ and $\gamma \preceq \alpha_i$.
We will index entries in the marginals by pairs $(i,\gamma)$,
where $\gamma \preceq \alpha_i$.
Let the total number of entries in marginals $C^{\alpha_1}, \dots ,
C^{\alpha_k}$ be equal to $\sum_{i = 1}^\ell 2^{\|\alpha_i\|} = K$.
Let $R$ be the recovery matrix for the Fourier strategy: $R \in \mathbb R^{K \times m}$ with entries
$R_{(i,\gamma),\alpha} = (C^{\alpha_i} f^{\alpha})_{\gamma}$. Then
$(C^{\alpha_1}, \dots, C^{\alpha_\ell}) = R \cdot (\hat f_1, \dots, \hat f_m)$.
Suppose that we are given a set of inconsistent noisy values of these
marginals $(\tilde C^{\alpha_1}, \dots, \tilde C^{\alpha_\ell})$.
We formulate the following optimization problem to find
the consistent set of marginals $(\bar C^{\alpha_1}, \dots, \bar
C^{\alpha_\ell})$ that is closest to the noisy values in $L^p$-norm:
\begin{align*}
\text{Minimize }
\|(\bar C^{\alpha_1}, \dots, \bar C^{\alpha_k}) -  (\tilde C^{\alpha_1}, \dots, \tilde C^{\alpha_\ell})\|_p \\
\text{ Subject to }(\bar C^{\alpha_1}, \dots, \bar C^{\alpha_\ell}) = R \cdot (\hat f_1, \dots, \hat f_m)
\end{align*}
For $p = 2$ this is gives a least squares problem, with $m$ variables
and $K$ constraints, which is expressed as:
$$\text{Minimize } \|R \cdot (\hat f_1, \dots, \hat f_m) - (\tilde C^{\alpha_1}, \dots, \tilde C^{\alpha_\ell})\|_2.$$
For $p = 1$ and $p = \infty$, this gives an LP similar to~\cite{BCDKST07}.

The running time of this consistency step via least squares only depends on the number of queries.
For example, for the case of all $k$-way marginals,
  we need to work with matrices of size $O(d^k)$, and perform a
  constant number of multiplications and inversions.
In contrast, prior work required solving LPs of size proportional to
the size of the data, $N=2^d$, which takes time polynomial in $N$.

\ifnum\mydebug=0
\newlength{\figwidth}
\setlength{\figwidth}{0.3\textwidth}

\begin{figure*}[t]
\subfigure[ADULT data, $Q_1$]{
 \label{adult1m}
  \includegraphics[width=\figwidth]{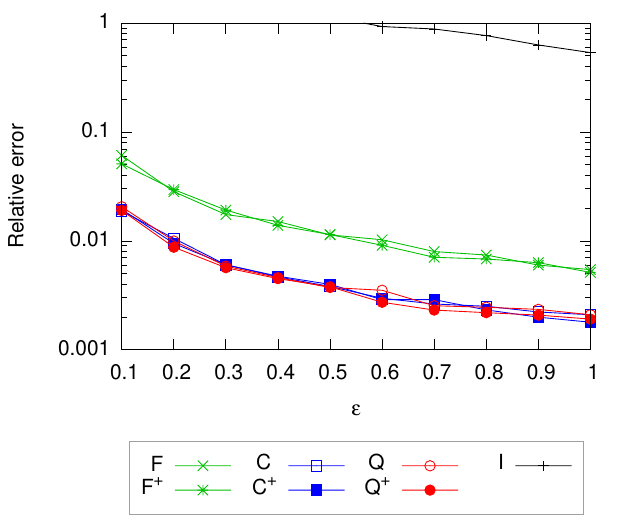}
}%
\subfigure[ADULT data, $Q_1^*$]{
 \label{adult1r}
  \includegraphics[width=\figwidth]{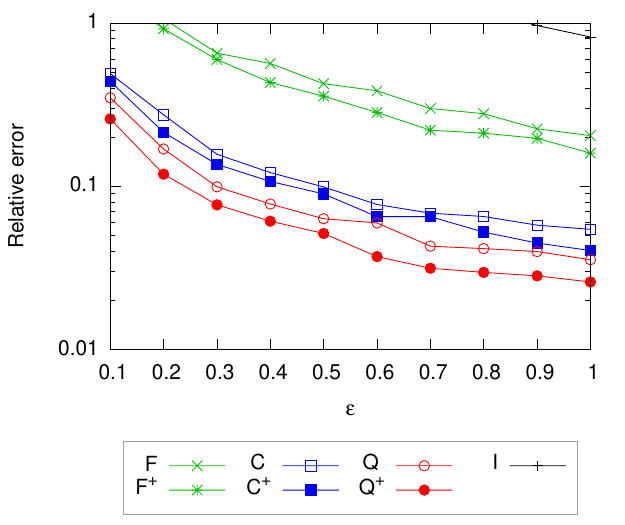}
}%
\subfigure[ADULT data, $Q_1^a$]{
 \label{adult1a}
  \includegraphics[width=\figwidth]{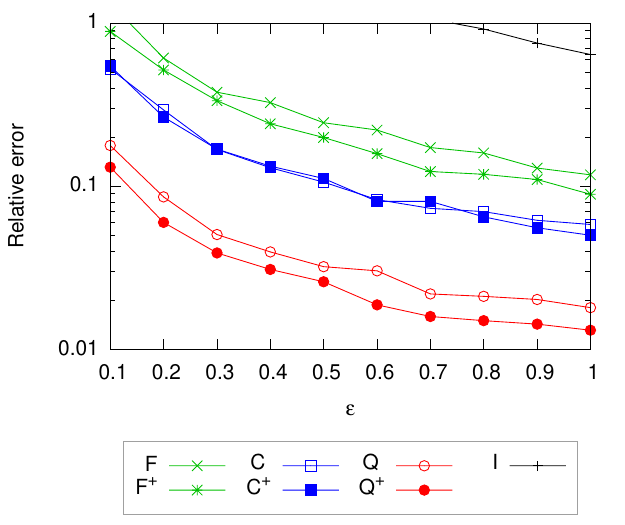}
}
\subfigure[ADULT data, $Q_2$]{
 \label{adult2m}
  \includegraphics[width=\figwidth]{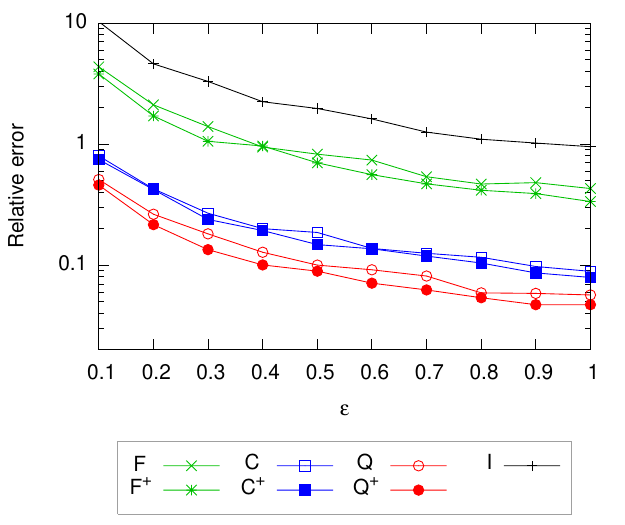}
}%
\subfigure[ADULT data, $Q_2^*$]{
 \label{adult2r}
  \includegraphics[width=\figwidth]{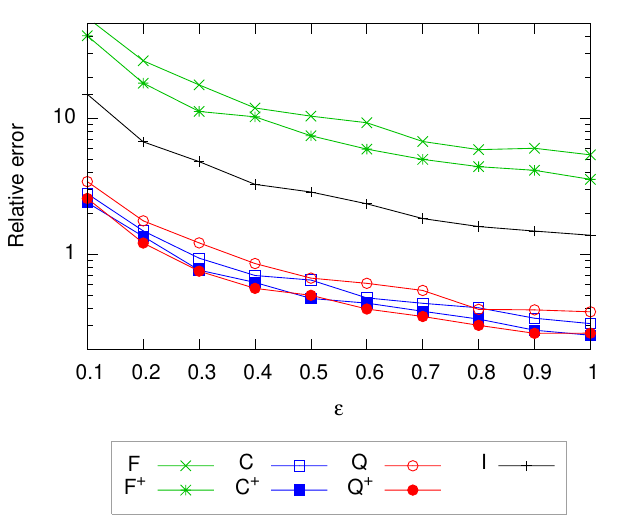}
}%
\subfigure[ADULT data, $Q_2^a$]{
  \label{adult2a}
  \includegraphics[width=\figwidth]{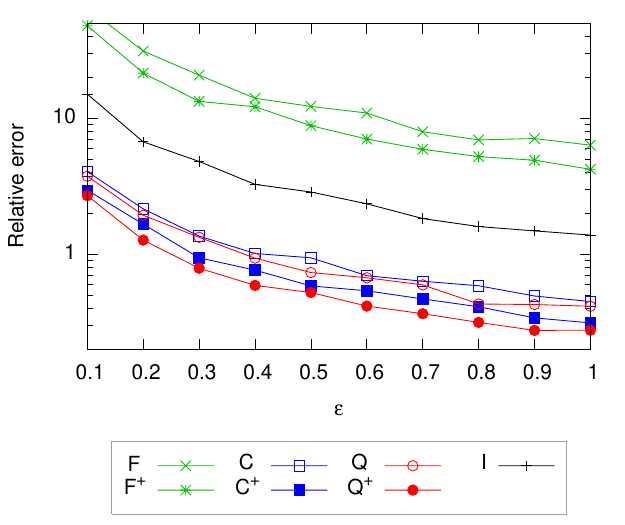}
}
\caption{Accuracy of marginal release on ADULT data}
\label{fig:adult}
\end{figure*}

\begin{figure*}[t]
\subfigure[NLTCS data, $Q_1$]{
 \label{nltcs1m}
  \includegraphics[width=\figwidth]{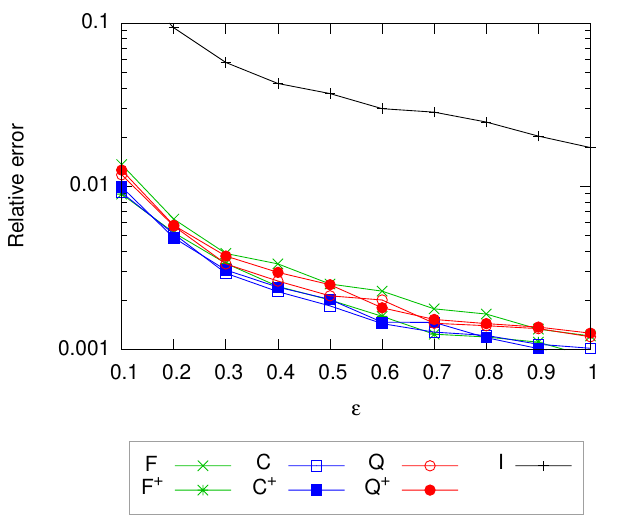}
}%
\subfigure[NLTCS data, $Q_1^*$]{
 \label{nltcs1r}
  \includegraphics[width=\figwidth]{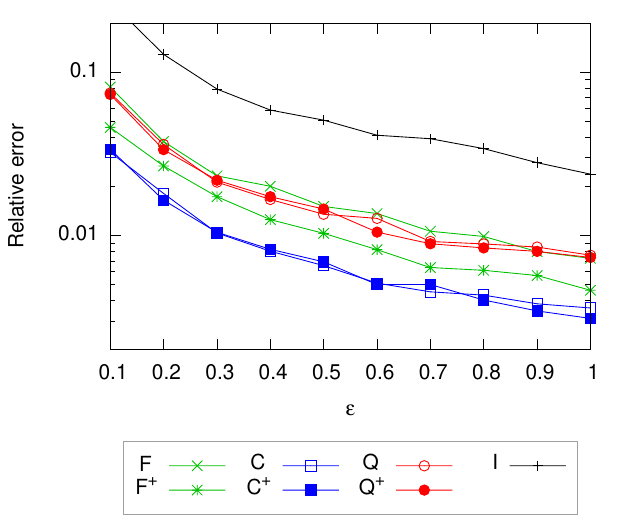}
}%
\subfigure[NLTCS data, $Q_1^a$]{
 \label{nltcs1a}
  \includegraphics[width=\figwidth]{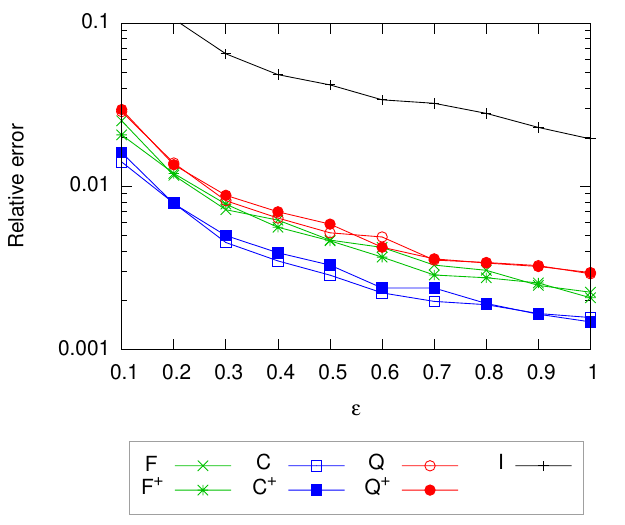}
}
\subfigure[NLTCS data, $Q_2$]{
 \label{nltcs2m}
  \includegraphics[width=\figwidth]{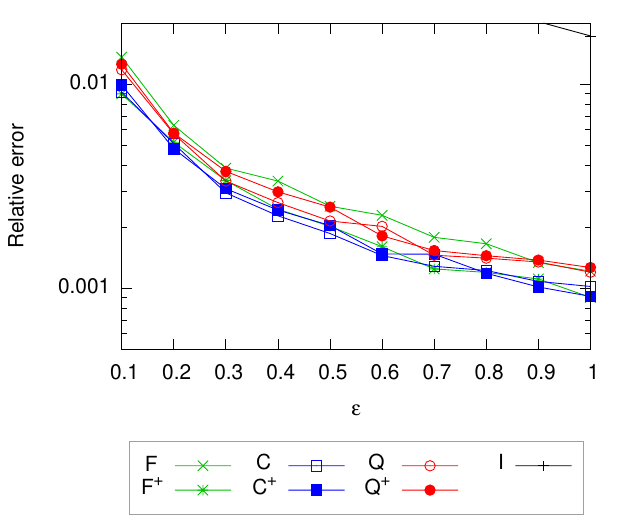}
}%
\subfigure[NLTCS data, $Q_2^*$]{
 \label{nltcs2r}
  \includegraphics[width=\figwidth]{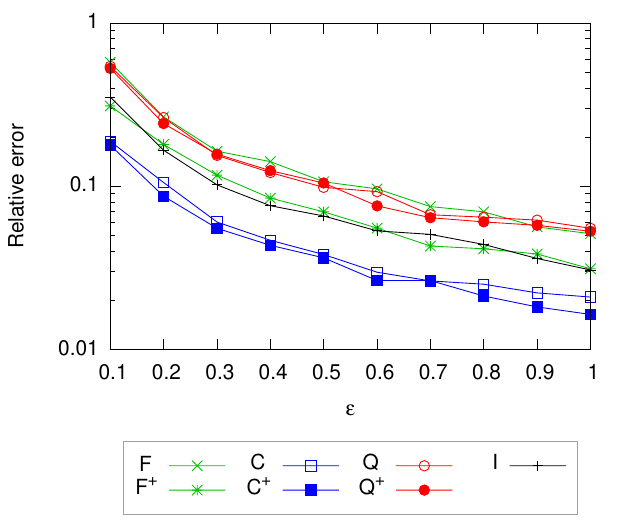}
}%
\subfigure[NLTCS data, $Q_2^a$]{
  \label{nltcs2a}
  \includegraphics[width=\figwidth]{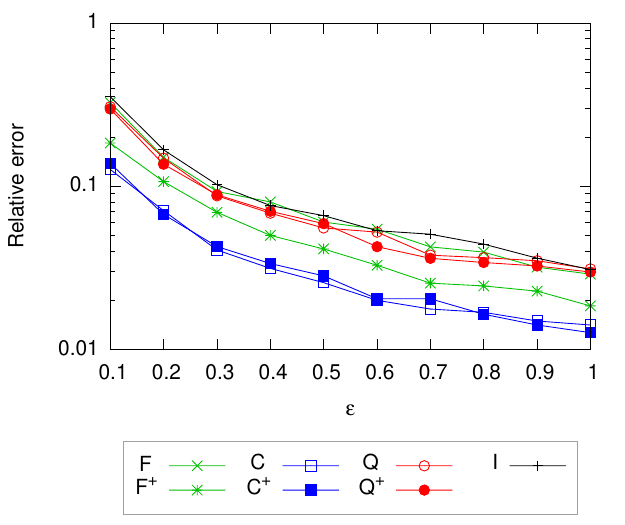}
}
\caption{Accuracy of marginal release on NLTCS data}
\label{fig:nltcs}
\end{figure*}

\section{Experimental study}\label{sec:experiments}
\label{expts}

\noindent
\textbf{Datasets.} We studied performance on two real datasets: \\
{\em Adult}: The Adult dataset from the UCI Machine Learning repository
  (\url{http://archive.ics.uci.edu/ml/}) has census information
on 32561 individuals. As in~\cite{DWHL11}, we extract
a subset of sensitive categorical attributes, for workclass
(cardinality 9), education (16), marital-status (7), occupation (15),
relationship (6), race (5), sex (2) and salary (2). \\
{\em NLTCS}: The National Long-Term Case Study 
from StatLib  (\url{http://lib.stat.cmu.edu/}), contains information about
  21576 individuals.
Each record consists 16 binary attributes, which correspond to
functional disability measures: 6 activities of daily living and 10
instrumental activities of daily living.

\medskip
\noindent
\textbf{Query workloads.}
The choice of query workloads in our experimental study is motivated
by an application of low-order marginals to statistical model
fitting.
In this setting, the typical set of queries consists of all $k$-way marginals
(for some small value of $k$) together with some subset of
$(k + 1)$-way marginals, chosen depending on the application.
We consider three different approaches: \\
1.  $Q_k$: all the $k$-way marginal tables.\\
2. $Q^*_k$: all the $k$-way marginal tables, plus half of all
  $(k + 1)$-way marginals. \\
3.  $Q^a_k$: all the $k$-way marginal tables, plus all $(k + 1)$-way
  marginals that include a fixed attribute.

\medskip
\noindent
\textbf{Evaluation metrics.}
We measure the average absolute error per entry in the set of
marginal queries.
To show the utility of these results, we scale each error by the
mean true answer of its respective marginal query, i.e., we plot it as a relative error.
Thus, a relative error of less than 1 is desirable, as otherwise the
true answers are dwarfed by the noise (on average).
Note that while the number of tuples in each dataset
is relatively small, our approaches do not depend on the tuple count, but rather
the dimensionality of the domain $N=\Pi_{i=1}^d |A_i|.$ Larger datasets would
only improve the quality metrics, while keeping
the running time essentially unchanged.

\medskip
\noindent
\textbf{Algorithms Used.}
We present results for
$\epsilon$-differential privacy.
Results for $(\epsilon, \delta)$-differential privacy are similar, and
are omitted.
We include seven approaches within the strategy/recovery framework,
based on choice of the strategy matrix, $S$.
Here, the notation $S^+$ indicates that we use the non-uniform noise
allocation for strategy $S$ as described in
Section~\ref{noisebudget}, while the corresponding $S$ is with uniform
noise.
\begin{itemize}
\item $S=I$ --- Add noise via Laplace mechanism directly to base
  cells and aggregate up to compute the marginals.
  Here, the optimal noise allocation is always uniform.
\item $S = Q$ --- Add uniform ($Q$) or non-uniform ($Q^+$) noise
  to each marginal independently.
\item $S=F$ --- Add uniform ($F$) or non-uniform ($F^+$) noise
  to each Fourier coefficient, corresponding to the given query workload.
\item $S=C$ --- Add uniform ($C$) or non-uniform ($C^+$) noise
  to each marginal returned by the greedy clustering strategy proposed in~\cite{DWHL11}.
\end{itemize}

Our goal is to study the effects of non-uniform noise budgeting over all strategies. The decision on which strategy to use rests with the data owner. However, we show clear tradeoffs between running time and accuracy for all strategies, which can provide helpful hints.

To ensure consistency of
the released marginals, we use the Fourier analytic approach,
described in Section~\ref{sec:consistency}.

\subsection{Adult Dataset}

Figure~\ref{fig:adult} shows the results on the Adult
data set, for query workloads $Q_1, Q_1^*, Q_1^a, Q_2, Q_2^*$ and $Q_2^a$.
The attributes in this data set have
varying cardinalities, but are encoded as binary attributes
as described in Section~\ref{sec:fourierdef}.
We plot the results on a logarithmic scale as we vary the privacy
parameter $\epsilon$, to more clearly show the relative performance of
the different measures.
Immediately, we can make several observations about the relative
performance of the different methods.
On this data, the naive method of materializing counts ($I$) is never
effective: the noise added is comparable to the magnitude of the data
in all cases.
Across the different query workloads, choosing the
strategy $S=Q$ works generally well. In this case non-uniform noise
allocation can significantly improve the accuracy.
For example, over workload $Q_1^*$ (Figure~\ref{adult1r}), we see an
improvement of 20-25\% in accuracy.

For more complex queries which result in more marginals of higher
degree ($Q_2^*$ and $Q_2^a$, in Figures~\ref{adult2r}
and \ref{adult2a} respectively), the accuracy is lower overall, and
the noise is greater than the magnitude of the data for more
restrictive settings of the privacy parameter $\epsilon$.
To some extent this can be mitigated as the number of individuals
represented in the data increases: the noise stays constant as the
value of the counts in the table grows.

Across this data, we observe that while the non-uniform approach
improves the accuracy of the Fourier strategy, it is inferior to other
strategies.
Although asymptotically this strategy has good properties (as
described in Table~\ref{tab:noise-k-way-summary}), in this case $k$ is
not very large, so the gap between the $k$ and $2^k$ terms for
constant $k$ is absorbed within the big-Oh notation.
The running times of our methods were all fast:
the Fourier ($F$) and Query ($Q$) methods took at most tens of seconds
to complete, while the clustering ($C$) took longer, due to
the more expensive clustering step.

\begin{figure}[t]
\centering
\includegraphics[width=0.7\columnwidth]{\plots/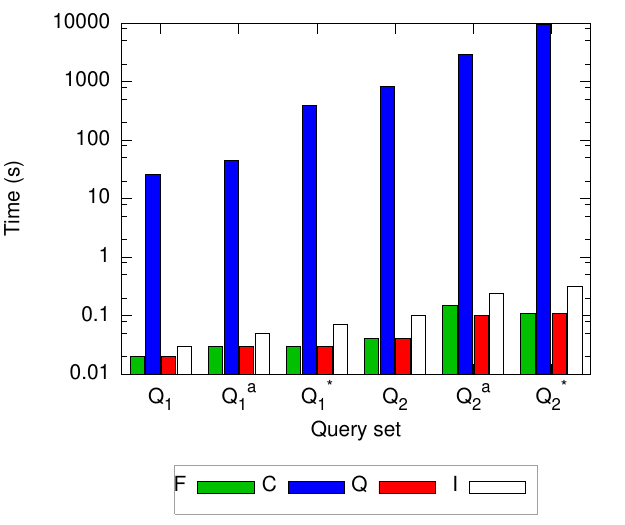}
\caption{Running time over NLTCS data}
\label{fig:nltcstime}
\vspace*{-6mm}
\end{figure}

\subsection{NLTCS data}

Figure~\ref{fig:nltcs} shows the corresponding results on the binary
NLTCS data.
Over all experiments, there is an appreciable benefit to applying
the optimal non-uniform budgeting.
The optimal budgeting case is reliably better than the uniform
version, for the same strategy matrix.
There are occasional inversions, due to the random nature of the
mechanisms used, but the trend is clear.
The advantage can be notable:
for example, on $Q_1^*$ (Figure~\ref{nltcs1r}) and $Q_2^*$
(Figure~\ref{nltcs2r}), the error of the Fourier strategy is reduced
30-35\% by using non-uniform budgeting.
For the clustering approach, the improvement
is smaller, but still measurable, around 5\% on
average.
However, recall that strategy $C$ becomes infeasible on higher
dimensional data, due to its exponential cost.

Figure~\ref{fig:nltcstime} shows the end-to-end running time of the
different methods.
This demonstrates clearly the dramatically slow running time of the
clustering method: reaching several hours to operate on a single,
moderate-sized dataset.
As the dimensionality increases, this becomes exponentially worse.
By contrast, the time needed by the other strategies is negligible:
always less than a second, and typically less than a tenth of a
second.
The optimization and consistency steps take essentially no time at
all, compared to the data handling and processing.
On the other hand, all methods have virtually constant time as a function of the number of tuples in the data.

Returning to the accuracy,
over the $k=1$-way marginals and variations, the approach of
materializing the base counts ($I$) is not competitive, while the
clustering strategy ($C$) achieves the least error.
The more lightweight Fourier-based approach achieves slightly more
error, but is much more scalable.
As the degree of the marginals increases ($Q_2^*$,
Figure~\ref{nltcs2r}, and $Q_2^a$, Figure~\ref{nltcs2a}), the
trivial solution of materializing the base cells becomes more
accurate.
For workloads that are made up of high-degree marginals, this method
dominates the other approaches, although such workloads are considered
less realistic.

\section{Concluding Remarks}
\label{concl}
We considered the problem of releasing data based on linear
queries, which captures the common case of data cubes and marginals.
We showed how existing matrix-based strategies can be improved by using
non-uniform noise based on the query workload.
Our results show that such non-uniform noise
results in significantly lower error across all cases considered.
Further, the cost of this is low, and the results can be made
consistent with minimal extra effort.

Other notions of consistency are possible within this framework.
For example, it is sometimes required that the query answers
correspond to a data set in which all counts are integral and non-negative.
This can be achieved when the method actually materializes a noisy set
of base counts $\hat{x}$ (as in the case of strategy $I$) by adding
the constraints that $\hat{x}_j \geq 0$ and rounding the results to
the nearest integer.
It remains to show how to enforce such consistency constraints
efficiently when base counts are not explicitly materialized.

On the theoretical side, we have shown bounds on accurate
$k$-way release under differential privacy of
$O(\frac{k}{\eps} \sqrt{\binom{d}{k}\binom{d+k}{k}})$.
An open problem is to close the gap
between this and the lower bound of
$\tilde{\Omega}(\frac{1}{\eps} \sqrt{\binom{d}{k}})$.

\fi

\medskip
\noindent
{\bf Acknowledgments.}
We  thank Adam D. Smith and Moritz Hardt for multiple useful
comments and suggestions.


\begin{thebibliography}{10}

\bibitem{BCDKST07}
B.~Barak, K.~Chaudhuri, C.~Dwork, S.~Kale, F.~McSherry, and K.~Talwar.
\newblock Privacy, accuracy, and consistency too: a holistic solution to
  contingency table release.
\newblock In {\em PODS}, 2007.

\bibitem{BLR08}
A.~Blum, K.~Ligett, and A.~Roth.
\newblock A learning theory approach to non-interactive database privacy.
\newblock In {\em STOC}, 2008. 

\bibitem{CSS10}
T.-H.~H. Chan, E.~Shi, and D.~Song.
\newblock Private and continual release of statistics.
\newblock In {\em  ICALP}, 2010. 

\bibitem{CPSSY12}
G.~Cormode, C.~Procopiuc, E.~Shen, D.~Srivastava, and T.~Yu.
\newblock Differentially private spatial decompositions.
\newblock In {\em ICDE}, 2012.

\bibitem{CPST12}
G.~Cormode, C.~Procopiuc, D.~Srivastava, and T.~Tran.
\newblock Differentially private publication of sparse data.
\newblock In {\em ICDT}, 2012.

\bibitem{DWHL11}
B.~Ding, M.~Winslett, J.~Han, and Z.~Li.
\newblock Differentially private data cubes: optimizing noise sources and
  consistency.
\newblock In {\em SIGMOD}, 2011.

\bibitem{Dwork:06}
C.~Dwork.
\newblock Differential privacy.
\newblock In {\em ICALP}, 2006.

\bibitem{DKMMN06}
C.~Dwork, K.~Kenthapadi, F.~McSherry, I.~Mironov, and M.~Naor.
\newblock Our data, ourselves: Privacy via distributed noise generation.
\newblock In {\em EUROCRYPT}, 2006.

\bibitem{DMNS06}
C.~Dwork, F.~McSherry, K.~Nissim, and A.~Smith.
\newblock Calibrating noise to sensitivity in private data analysis.
\newblock In {\em Theory of Cryptography}, 2006.

\bibitem{GRU12}
A.~Gupta, A.~Roth, and J.~Ullman.
\newblock Iterative constructions and private data release.
\newblock {\em CoRR}, abs/1107.3731, 2011.

\bibitem{HLS10}
M.~Hardt, K.~Ligett, and F.~McSherry.
\newblock A simple and practical algorithm for differentially private data
  release.
\newblock {\em CoRR}, abs/1012.4763, 2010.

\bibitem{HR10}
M.~Hardt and G.~N. Rothblum.
\newblock A multiplicative weights mechanism for privacy-preserving data
  analysis.
\newblock In {\em FOCS}, 2010

\bibitem{HT10}
M.~Hardt and K.~Talwar.
\newblock On the geometry of differential privacy.
\newblock In {\em  STOC}, 2010. 

\bibitem{HRMS10}
M.~Hay, V.~Rastogi, G.~Miklau, and D.~Suciu.
\newblock Boosting the accuracy of differentially-private histograms through
  consistency.
\newblock In {\em VLDB}, 2010.

\bibitem{KRSU10}
S.~P. Kasiviswanathan, M.~Rudelson, A.~Smith, and J.~Ullman.
\newblock The price of privately releasing contingency tables and the spectra
  of random matrices with correlated rows.
\newblock In {\em STOC}, 2010.

\bibitem{LHRMG10}
C.~Li, M.~Hay, V.~Rastogi, G.~Miklau, and A.~McGregor.
\newblock Optimizing linear counting queries under differential privacy.
\newblock In {\em PODS}, 2010. 

\bibitem{LM11}
C.~Li and G.~Miklau.
\newblock Efficient batch query answering under differential privacy.
\newblock {\em CoRR}, abs/1103.1367, 2011.

\bibitem{LZWY11}
Y.~Li, Z.~Zhang, M.~Winslett, and Y.~Yang.
\newblock Compressive mechanism: Utilizing sparse representation in
  differential privacy.
\newblock In {\em WPES}, 2011.

\bibitem{MM09}
F.~McSherry and I.~Mironov.
\newblock Differentially private recommender systems: Building privacy into the
  netflix prize contenders.
\newblock In {\em KDD}, 2009.

\bibitem{Rao:65}
C.~R. Rao.
\newblock {\em Linear statistical inference and its applications}.
\newblock Wiley, 1965.

\bibitem{BV04}
S.~Boyd and L.~Vandenberghe.
\newblock Convex Optimization.
\newblock Cambridge Univ. Press, 2004.

\bibitem{RR10}
A.~Roth and T.~Roughgarden.
\newblock Interactive privacy via the median mechanism.
\newblock In {\em STOC}, 2010.

\bibitem{XWG10}
X.~Xiao, G.~Wang, and J.~Gehrke.
\newblock Differential privacy via wavelet transforms.
\newblock In {\em ICDE}, 2010.

\end{thebibliography}

\ifnum\final=0
\newpage
\appendix

\section{Budgeting for $(\epsilon, \delta)$-differential privacy}\label{app:approx-dp-budgeting}

\begin{proofof}{Proposition~\ref{prp:nonuniform-budgeting}, part (ii)}
Consider a matrix $S' \in \mathbb R^{m \times N}$ with entries $S'_{ij} = \frac{\epsilon_i S_{ij}}{\alpha}$, where $\alpha = \max_{j = 1}^N \sqrt{\sum_{i = 1}^m \epsilon^2 S^2_{ij}}$.
Because $L_2$-norm of any column of $S'$ is at most $1$ by Theorem~\ref{thm:gaussian-mechanism} adding to values $S'_i x$ Gaussian noise with variance $\frac{2 \log(2/\delta)}{\epsilon^2}$ guarantees $(\epsilon, \delta)$-differential privacy. Given these noisy values which we denote as $\tilde S'_i$, the noisy values $\tilde S_i$ can be computed as $\tilde S_i = \frac{\alpha \tilde S'_i}{\epsilon_i}$ and so $\tilde S_i$ has variance $\frac{2 \log(2/\delta) \alpha^2}{\epsilon_i^2 \epsilon^2}$, which is at most 
$\frac{2 \log(2/\delta)}{\epsilon_i^2}$ for $\epsilon > \alpha$ as desired.
\end{proofof}

\begin{proofof}{Corollary~\ref{cor:noisebudget-variance} for $(\epsilon, \delta)$-differential privacy}
For $(\epsilon,\delta)$-differential privacy the analog of
optimization problem~\eqref{soptpb} -- \eqref{ineq:eta} is:
\begin{align*}
\text{Minimize: } 2 \log(2/\delta) \sum_{i = 1}^m \frac{s_i}{\epsilon^2_i} \\
\sum_{i = 1}^m C^2_i \epsilon^2_i = \epsilon^2 
\end{align*}
We can ignore the multiplicative factor in the objective function, because it doesn't change the optimum solution. Then the corresponding Lagrange function is:
$$\Lambda(\lambda, \epsilon_i^2) = \left(\sum_{i = 1}^m \frac{s_i}{\epsilon^2_i}\right) +\lambda \left(\sum_{i = 1}^m C^2_i \epsilon^2_i - \epsilon^2\right).$$
Using condition $\frac{\partial}{\partial \epsilon^2_i} = 0$ we have:
$$-\frac{s_i}{\epsilon^4_i} + \lambda C^2_i = 0,$$
so $\epsilon^2_i = \frac{1}{C_i} \sqrt{\frac{s_i}{\lambda}}$ and we have
$$\sum_{i = 1}^m C_i \sqrt{\frac{s_i}{\lambda}} = \epsilon^2,$$
which gives $\sqrt{\lambda} = \frac{1}{\epsilon^2} \sum_{i = 1}^m C_i \sqrt{s_i}$ and $\epsilon^2_i = \frac{\epsilon^2}{C_i} \frac{\sqrt{s_i}}{\sum_{j = 1}^m C_j \sqrt{s_j}}$.
The value of the objective function is now given as:
$$\frac{2 \log(2/\delta)}{\epsilon^2} \left(\sum_{i = 1}^m C_i \sqrt{s_i}\right)^2.$$
If all values $C_i$ are equal to $C$ this gives $\frac{2 \log(2/\delta)C^2}{\epsilon^2}\left( \sum_{i = 1}^m \sqrt{s_i}\right)^2$.
\end{proofof}

\section{Fourier strategy--Uniform Noise}
\label{app:marg}
In this section, we provide a tighter analysis of the expected noise
that results from using the Fourier strategy with uniform noise. 
This is to be compared with the bound of 
$O(2^{\|\alpha\|} |B| \log(|B|/\delta)/\eps)$ stated
in \cite[Theorem~7]{BCDKST07}. 

\begin{theorem}\label{thm:improved-fourier-noise}
Let the query workload $Q$ consist of marginals $C^{\alpha_1}, \dots,
C^{\alpha_k}$ and $B$ be the set of Fourier coefficients,
corresponding to this workload, such that 
$B = \{\beta | \exists C^{\alpha_i} \colon \beta \preceq \alpha_i\}$. 
Then if we release all Fourier coefficients in $B$ via the Laplace
mechanism with uniform noise and use them to compute private values of
the marginals $\tilde C^{\alpha_1}, \dots, \tilde C^{\alpha_k}$ for
each marginal $C^{\alpha_i}$ bounds on the noise per marginal can be given
as: 
\begin{enumerate}
\item $\E[\|C^{\alpha_i} x - \tilde C^{\alpha_i}\|_1] \le \frac{|B| \sqrt{2^{3+\|\alpha_i\|}}}{\epsilon}.$
\item $\|C^{\alpha_i} x - \tilde C^{\alpha_i}\|_1 =
  O\left(\frac{|B|\sqrt{2^{\|\alpha_i\|}
  \log\left(|B|/\delta\right)}}{\epsilon}\right),$ 
with probability at least $1 - \delta.$
\end{enumerate}
\end{theorem}

\begin{proof}

Let $\alpha = \alpha_i$ and w.l.o.g, assume that $x = 0$, so:
\begin{align}
\|C^\alpha x - \tilde C^\alpha \|_1 & = 
\|\tilde C^\alpha \|_1  = 
\bigg\|\sum_{\beta \preceq \alpha} \phi_\beta C^{\alpha} f^{\beta}\bigg\|_1 
\nonumber \\
 & = \sum\limits_{\gamma \preceq \alpha} \bigg|\sum_{\beta \preceq \alpha} \phi_\beta (C^{\alpha} f^{\beta})_\gamma\bigg| \nonumber\\
& = \sum\limits_{\gamma \preceq \alpha} \bigg|\sum_{\beta \preceq \alpha} \phi_\beta \sum_{\eta \colon \eta \wedge \alpha = \gamma}(-1)^{\langle \beta, \eta \rangle}2^{-d/2}\bigg| \nonumber\\
&  = {2^{-d/2}}\sum\limits_{\gamma \preceq \alpha} \bigg|\sum_{\beta \preceq \alpha} \phi_\beta \cdot 2^{d - \|\alpha\|}(-1)^{\langle \beta, \gamma \rangle}\bigg| \nonumber\\
& = 2^{d/2}
\bigg|\sum_{\beta \preceq \alpha} \phi_\beta\bigg|
\label{eqn:l1-noise}
\end{align}

The last step follows due to the symmetry of the Laplace random
variables $\phi_\beta$. 
Note that for iid unbiased random variables $Y$, 
$\E[|Y|] =  \E[\sqrt{Y^2}] \leq  \sqrt{\E[Y^2]} = \sqrt{\Var(Y)}$,
using Jensen's inequality. 
Hence, using the linearity of expectation and the fact that 
$\Var(\phi_\beta) = \frac{8 |B|^2}{\epsilon^2 2^d}$, we have:
\begin{align}
&\E\left[\|C^\alpha x - \tilde C^{\alpha} \|_1\right] = 
\E\bigg[2^{d/2}
\bigg|\sum_{\beta \preceq \alpha}
\phi_\beta \bigg| \bigg] \nonumber \\
& \le 2^{d/2}
\sqrt{2^{\|\alpha\|}\Var(\phi_\beta)}
= 
\frac{\sqrt{2^{3+\|\alpha\|}} |B|}{\epsilon}
\label{eqn:noise-bound}.
\end{align}


To get a concentration bound for the second part of the Theorem,
we use the fact that with probability at least $1-\delta$, 
$|\sum_{\beta} \phi_\beta| = O(\sqrt{\sum_{\beta} \Var(\phi_\beta) \log
1/\delta})$ (see, e.g.~\cite{CSS10}).

Substituting this in \eqref{eqn:l1-noise}, 
we have with probability $1-\delta$, 
\[ \|C^{\alpha_i} x - \tilde C^{\alpha_i} \|_1 =
O\left(\sqrt{2^{\|\alpha_i\|}}\frac{|B|}{\eps} \log^{1/2} \frac{1}{\delta}\right) \] 

Rescaling $\delta$ by a factor of $|B|$ \grinote{Why union bound on $|B|$, not the number of marginals?} means that by a union bound 
this holds for all $\alpha_i$. 
\end{proof}

\section{Omitted Proofs}
\label{app:k-way}
\begin{proofof}{Lemma~\ref{lem:budgeting-k-way-marginals}}
We have
  $q = 2^k \binom{d}{k}$, $m = \sum_{i = 0}^k \binom{d}{i}$ and $N = 2^d$.
For the set of all $k$-way marginals the matrices 
 $Q \in \mathbb R^{q \times N}$, $S \in \mathbb R^{m \times N}$ and 
 $R \in \mathbb R^{q \times m}$ 
 have the following entries: 
\begin{align*}
&Q_{(i,t)j} = \begin{cases}
    1, \mbox{ if } i \wedge j = t \\
    0, \mbox{ otherwise,}
  \end{cases}
&
\end{align*}
\begin{align*}
&R_{(i,t)j} = \begin{cases}
    (-1)^{\langle i, t\rangle} 2^{d/2 - k}, \mbox{ if } j \preceq i \\
    0, \mbox{ otherwise.}
  \end{cases}
&S_{ij} = (-1)^{\langle i, j \rangle} / 2^{d/2}, 
\end{align*}
\noindent
where we abuse notation indexing entries by vectors $i,j,t \in
\{0,1\}^d$ only if entries corresponding to respective subsets exist. 
Here, $t$ indexes $\binom{d}{k}$ marginals. 
Formally, we use all pairs $(i,t)$, where $\|i\| = k$ and $t \preceq i$
to index over $[q]$.
We use an index $i$, where $\|i\| \le k$ to index over $[m]$ 
and a regular index $0 \le i < 2^d$ to index over $[N]$.
The grouping number of $S$ is equal to $m$ because the groups consist of individual rows of $S$, so for all $1 \le i \le m$ we have $C_i = 1 / 2^{d/2}$ and $b_i = 2 \sum_{j = 1}^q a_j R^2_{ji}$.
Using $a = \overrightarrow{\mathbf 1_q}$ and substituting entries of $R$, we have
$b_i = 2^{d - k + 1} \cdot \binom{d - \|i\|}{k - \|i\|}$, where we use index $i$ as described above.
Thus, $$\sum_{i = 1}^m b_i^{1/3} = 2^{(d - k + 1) / 3}\sum_{i = 0}^k \binom{d}{i} \binom{d - i}{k - i}^{1/3} .$$
%

Now, using Corrollary~\ref{cor:noisebudget-variance}
privacy (and assuming $k < d/2$) we get the sum of the noise variances over all entries equal to 
\begin{align*}
  & \frac{1}{2^{k - 1} \eps^2} \bigg(\sum_{i = 0}^k \binom{d}{i} \binom{d - i}{k - i}^{1/3}\bigg)^3
   \le \frac{3(k + 1)^2}{2^{k - 1} \eps^2} \sum_{i = 0}^k \binom{d}{i}^3 \binom{d - i}{k - i}
 \\ & = \frac{3 (k + 1)^2}{2^{k - 1} \eps^2} \sum_{i = 0}^k
  \binom{d}{i}^2 \binom{d}{k} \binom{k}{i} \\
&  = \frac{3 (k + 1)^2 \binom{d}{k}}{2^{k - 1} \eps^2} \sum_{i = 0}^k \binom{d}{i}^2  \binom{k}{i}
  \le \frac{3 (k + 1)^2 \binom{d}{k}^2}{2^{k - 1} \eps^2} \sum_{i = 0}^k \binom{d}{i}  \binom{k}{i}
  \\ & = \frac{3 (k + 1)^2 \binom{d + k}{k}\binom{d}{k}^2}{2^{k - 1} \eps^2}
\end{align*}

Thus, dividing by $q$, the variance of noise per entry of a marginal table is  $O\left(\frac{k^2 \binom{d}{k} \binom{d + k}{k}}{2^{2k} \eps^2}\right)$.
By Jensen's 
inequality, the expected magnitude of noise per entry is 
$O\left(\frac{k \sqrt{\binom{d}{k} \binom{d + k}{k}}}{2^{k} \eps}\right)$.

For $(\eps,\delta)$-differential privacy Corollary~\ref{cor:noisebudget-variance} gives variance at most: 
\begin{align*}
  O&\bigg(\frac{\log(1 / \delta)}{2^{k} \eps^2} \bigg(\sum_{i = 0}^k \binom{d}{i} \binom{d - i}{k - i}^{1/2}\bigg)^2\bigg)
\\ &  \le O\bigg(\frac{\log(1 / \delta)}{2^{k} \eps^2} \cdot 2 (k + 1) \sum_{i = 0}^k \binom{d}{i}^2 \binom{d - i}{k - i}\bigg) \\
  & = O\bigg(\frac{k \log(1 / \delta)}{2^{k} \eps^2} \sum_{i =
  0}^k \binom{d}{i} \binom{d}{k} \binom{k}{i}\bigg) \\
& = O\bigg(\frac{k \binom{d}{k} \binom{d + k}{k} \log(1 / \delta)}{2^{k} \eps^2}\bigg),
\end{align*}
which gives $O(\frac{1}{2^k\eps}\sqrt{k \binom{d + k}{k} \log(1 / \delta)})$ expected noise per marginal entry.
\end{proofof}

\fi
\end{document}